\documentclass{article}
\usepackage[margin=1in]{geometry}  % set the margins to 1in on all sides

\usepackage{listings}
\usepackage{amsmath,mathtools}
\usepackage{amsthm}
\usepackage{tikz}
\usetikzlibrary{positioning}
\usepackage{caption,subcaption}
\usepackage{array}
\usepackage{mdwmath}
\usepackage{multirow}
\usepackage{mdwtab}
\usepackage{eqparbox}
\usepackage{multicol}
\usepackage{amsfonts}
\usepackage{multirow,bigstrut,threeparttable}
\usepackage{amsthm}
\usepackage{array}
\usepackage{bbm}
\usepackage{subcaption}
\usepackage{epstopdf}
\usepackage{mdwmath}
\usepackage{mdwtab}
\usepackage{eqparbox}
\usepackage{latexsym}
\usepackage{amssymb}
\usepackage{bm}
\usepackage{amssymb}
\usepackage{graphicx}
\usepackage{mathrsfs}
\usepackage{epsfig}
\usepackage{psfrag}
\usepackage{setspace}
\usepackage{tikz-cd}

\usepackage[%dvips,
CJKbookmarks=true,
bookmarksnumbered=true,
bookmarksopen=true,
colorlinks=true,
citecolor=red,
linkcolor=blue,
anchorcolor=red,
urlcolor=blue
]{hyperref}
\usepackage{cleveref}
\usepackage{algorithm}
\usepackage{algorithmic}
\usepackage{stfloats}
\usepackage[
autocite    = superscript,
backend     = biber, % bibtex, biber. bibtex: warning: "Using fall-back BibTeX(8) backend:(biblatex) functionality may be reduced/unavailable."
sortcites   = true,
style       = authoryear % numeric; nature
]{biblatex} % to be deleted
\usepackage{nameref}
\usepackage[normalem]{ulem}

\usepackage{soul}

\input xy
\xyoption{all}

\theoremstyle{plain}
\newtheorem{theorem}{Theorem}[section]
\newtheorem{proposition}{Proposition}

\newtheorem{corollary}{Corollary}[section]

\newtheorem{assumption}{Assumption}[section]

\crefname{theorem}{Theorem}{Theorems}
\crefname{proposition}{Proposition}{Propositions}
\crefname{corollary}{Corollary}{Corollaries}

\theoremstyle{definition}

\crefname{definition}{Definition}{Definitions}
\crefname{remark}{Remark}{Remarks}

\def\EE{\mathbb{E}}

\def\PP{\mathbb{P}}
\def\calN{\mathcal{N}}

\def\1{\mathbbm{1}}
\def\var{\mathsf{Var}}

\usepackage{booktabs}
\usepackage{adjustbox}
\usepackage{multirow}
\usepackage{listings}

\theoremstyle{plain}

\def \var {\mathsf{Var}}

\usepackage{xspace}

\definecolor{myblue}{rgb}{.8, .8, 1}
\definecolor{mathblue}{rgb}{0.2472, 0.24, 0.6} % mathematica's Color[1, 1--3]
\definecolor{mathred}{rgb}{0.6, 0.24, 0.442893}
\definecolor{mathyellow}{rgb}{0.6, 0.547014, 0.24}

\usepackage{enumitem}

\usepackage{pgfplots}
\pgfplotsset{compat=1.17}

\newcommand{\IFC}{\text{IF}}

\title{Trustworthy assessment of heterogeneous treatment effect estimator via analysis of relative error}
\author{
Zijun Gao
  \footnote{\small Marshall School of Business, University of Southern California, USA}
}

\begin{document}

\maketitle

% TODO:
% 0. Proofs
% 1. Simulations
% 2. Compare >=3 HTE estimators.
% 3. Future directions

\begin{abstract}
    % We develop a powerful and robust method for comparing heterogeneous treatment effect estimators based on the semi-parametrically efficient estimator of relative errors.
    Accurate heterogeneous treatment effect (HTE) estimation is essential for personalized recommendations, making it important to evaluate and compare HTE estimators. Traditional assessment methods are inapplicable due to missing counterfactuals.
    Current HTE evaluation methods rely on additional estimation or matching on test data, often ignoring the uncertainty introduced and potentially leading to incorrect conclusions. 
    We propose incorporating uncertainty quantification into HTE estimator comparisons.
    In addition, we suggest shifting the focus to the estimation and inference of the relative error between methods rather than their absolute errors.
    Methodology-wise, we develop a relative error estimator based on the efficient influence function and establish its asymptotic distribution for inference. Compared to absolute error-based methods, the relative error estimator (1) is less sensitive to the error of nuisance function estimators, satisfying a "global double robustness" property, and (2) its confidence intervals are often narrower, making it more powerful for determining the more accurate HTE estimator. 
    Through extensive empirical study of the ACIC challenge benchmark datasets, we show that the relative error-based method more effectively identifies the better HTE estimator with statistical confidence, even with a moderately large test dataset or inaccurate nuisance estimators.
\end{abstract}

\section{Introduction}\label{sec:introduction}

The estimation of heterogeneous treatment effects (HTE) under the Neyman-Rubin potential outcome framework is becoming increasingly prominent, driven by the need for tailored approaches in areas such as personalized medicine, personalized education, and personalized advertising. 
A variety of machine learning tools are being employed to estimate HTE, including LASSO, random forests, gradient boosting, and neural networks. 
Despite a rich body of research on HTE estimation, evaluating and comparing these estimators remain less investigated.
% There are two primary reasons to focus on the assessment of HTE estimators. 
% Firstly, assessing an estimator's absolute accuracy predicts its performance on future data. 
% \zg{Goal: determine the better one.}
% \zg{Directly estimating the better one, and in particular, construct CI for that. Estimator is a simply difference, but CI is not...}
Evaluating the performance of HTE estimators is essential for identifying a better candidate, especially considering the wide range of HTE estimation methods available.

% including the degree of penalization in LASSO, tree sizes in random forests, number of layers in neural networks. 
One significant challenge in assessing an HTE estimator arises from the inherent missingness in potential outcome model.
Provided with a test dataset, standard evaluation of a predictor compares the actual observations with the predicted values.
% This is feasible because the observations are unbiased samples of the predicted values. 
However, in the potential outcome model, each observed response corresponds to one potential outcome (treatment or control), and the HTE---the difference between two potential outcomes---is not directly observed.
To deal with the missingness of HTE, existing methods of HTE assessment typically involves additional steps performed on the test dataset, such as matching \parencite{rolling2014model} or nuisance function estimation \parencite{alaa2019validating}, to construct pseudo-observations of the HTE (\Cref{sec:literature}). 
The additional steps could introduce substantial randomness, which may even dominate the actual error difference between two HTE estimators.
% where the randomness of the error exceeds the magnitude of the gap of two HTE estimators.
Simply ignoring this source of randomness and outputting the estimator with the lower point error estimate could result in incorrect decisions with a significant probability.
% This observation motivates us to account for the uncertainty of the estimated errors in model comparison. 

% Among existing methods of HTE assessment, one popular method considers the prediction of the observed outcomes \zg{Ref}.
% Accurate predictions of both the treated and the control potential outcomes are sufficient for a precise estimation of the treatment effect, but not necessary.
% In addition, there are approaches, like various \zg{tree-based methods of HTE}, yielding only the estimates of HTE but not the potential outcomes, and the outcome-prediction-based assessment method can not be applied.
% Another line of methods first match treated and control units to create ``virtual twins'', and use the difference between the paired units' observed outcomes as pseudo-observations of the HTE \zg{Ref}. Nevertheless, the matching can be computationally intensive, and less viable when there is a significant imbalance (leaving part of the units unmatched and thus not used) between the treatment arm and the control arm. 

% \subsection{Proposal overview and contributions}\label{sec:overview}

In this paper, we advocate that the comparison of HTE estimators should account for the randomness introduced during the evaluation stage. 
Rather than providing a point estimate of the error, we suggest constructing a confidence interval for the evaluation error, which contains the true error value with a pre-specified probability.
We then demonstrate that the confidence interval for the absolute error of an HTE estimator (1) can be sensitive to nuisance function estimation on the test data; (2) for two similar HTE estimators, their absolute error confidence intervals do not account for the similarity of the HTE estimators and could be too wide to determine the more accurate estimator.
To address the issues of uncertainty quantification for absolute errors, we propose to directly construct confidence intervals for the \textit{relative} error between two HTE estimators, rather than their individual absolute errors. 
Methodologically, we derive an efficient estimator for the relative error using influence functions and characterize its asymptotic distribution to facilitate the confidence interval construction. 
Theoretically, we prove that our confidence interval of relative errors is valid under weaker assumptions regarding the quality of the nuisance function estimators compared to that of absolute errors, and is guaranteed to be narrow when the HTE estimators for comparison are similar.
Empirically, we show that the relative error confidence intervals achieves better coverage as well as are more powerful in identifying the better HTE estimator.
Beyond the difference in conditional means, our proposal can also be generalized to compare estimators of a broader class of heterogeneous treatment effect estimands more suitable for quantifying treatment effects for non-continuous outcomes.

% \noindent\textbf{Contributions}. 
% Our contributions are three-fold:
% \begin{itemize}
%     \item [1.] For comparing HTE estimators, we propose to account for the uncertainty in estimating the error of HTE estimators, which is largely ignored in the literature. Taking the uncertainty into consideration yields more trustworthy conclusion about the quality comparison of HTE estimators.
%     \item [2.] We suggest constructing confidence intervals of the relative error between two HTE estimators rather than their absolute errors. 
%     We propose a one-step correction estimator and the associated confidence interval for the relative error based on the efficient influence function, and prove the asymptotic validity and optimality of the proposed confidence interval. 
%     We prove the relative error estimator is less sensitive to nuisance function estimators, enjoys a global doubly robustness property, and is useful in selecting the better HTE estimator.
% \end{itemize}
% On the simulated data, we observe the constructed interval of the relative error achieves the desirable coverage corresponding to controlling the probability of making incorrect conclusions.

\noindent\textbf{Organization}.
The paper is organized as follows. 
In \Cref{sec:background}, we present the problem formulation, introducing the potential outcome model, specifying the causal estimand, defining the absolute and relative errors for HTE estimators, and providing a review of the relevant literature. In \Cref{sec:absolute.error}, we provide the efficient estimator, confidence interval of absolute errors, and discuss the issues therein.
In \Cref{sec:relative.error}, we focus on relative errors, presenting the efficient estimator and its associated confidence interval, and explain how the issues with absolute error confidence intervals are avoided.
We also discuss the generalization of our proposal based on relative error to a broader class of causal estimands. 
In \Cref{sec:simulation}, we compare the confidence intervals for absolute and relative errors on a benchmark dataset from the 2016 ACIC challenge. 
In \Cref{sec:discussion}, we conclude with a summary and outline future research directions. 
All proofs are provided in the supplementary materials.

% \zg{Add figures of uncertainty quantification to the main text.}

\section{Formulation and background}\label{sec:background}

\subsection{Potential outcome model}\label{sec:potential.outcome.model}

We follow the Neyman-Rubin potential outcome model with a treatment condition and a control condition.
% Suppose there are $n$ units.
For unit $i$, there is a $d$-dimensional covariate vector $X_i$, a binary treatment assignment $W_i \in \{0, 1\}$, and two potential outcomes: $Y_i(0)$ under the control condition and $Y_i(1)$ under the treatment condition.
We denote the observed outcome by $Y_i$, which equals $Y_i(1)$ if the unit is under treatment, and $Y_i(0)$, otherwise.
We use $Z_i = (X_i, W_i, Y_i)$ to denote all the observed data of unit $i$.

We make the conventional assumptions of SUTVA, overlap, and unconfoundedness.
\begin{assumption}[Stable unit treatment value assumption (SUTVA)]
    The potential outcomes for any unit do not depend on the treatments assigned to other units. There are no different versions of each treatment level.
\end{assumption}

\begin{assumption}[Unconfoundedness]
    The assignment mechanism does not depend on potential outcomes:
    \begin{align*}
        (Y_i(1), Y_i(0)) \perp W_i \mid X_i.
    \end{align*}
\end{assumption}

\begin{assumption}[Overlap]\label{assu:overlap}
    There is a positive probability of receiving treatment and control for all individuals.
\end{assumption}

To define our causal estimand, we introduce the super-population. 
We assume individuals are sampled i.i.d. from a super-population, denoted by $\PP$.
In particular, the covariates $X_i$ are sampled from an unknown distribution $\PP_X$. 
Given the covariates $X_i$, a binary group assignment $W_i \in \{0, 1\}$ is generated from a Bernoulli distribution with mean $e(X_i)$ (also known as the propensity score). 
Echoing with \Cref{assu:overlap}, we assume there exists $\eta > 0$  such that $\eta < e(x) < 1 - \eta$.
The potential outcomes are modeled by
\begin{align*}
    Y_i(1) | X_i &= \mu_1(X_i) + \epsilon_i, \\
    Y_i(0) | X_i &= \mu_0(X_i) + \epsilon_i,  
\end{align*}
where $\mu_0(x)$, $\mu_1(x)$ represent the conditional mean function of the control and the treatment group respectively, and $\epsilon_i$ denotes the error term, assumed to be zero-mean and independent of $X_i$. 
The estimand HTE is defined as 
\begin{align*}
    \tau(x) = \EE\left[Y(1) - Y(0) \mid X = x\right] 
    = \mu_1(x) - \mu_0(x),
\end{align*}
which can also be expressed as the difference of the two group conditional mean functions.

\subsection{Absolute error and relative error}\label{sec:error.definition}

In this paper, we focus on the evaluation and comparison of the HTE estimators using a test dataset of size $n$ drawn from the super-population $\PP$.
We use $\hat{\tau}_1(x)$, $\hat{\tau}_2(x)$ to denote HTE estimators derived independently of the test dataset.
When there is only one HTE estimator, we drop the subscript and use $\hat{\tau}(x)$.
We highlight that the HTE estimators are provided to us and the HTE estimation problem itself is not the focus of this paper.

To quantify the accuracy of an HTE estimator $\hat{\tau}(x)$, the absolute error is defined as
\begin{align}\label{defi:evaluation.error}
    \phi(\hat{\tau}(x))
    := \EE\left[(\hat{\tau}(X) - \tau(X))^2\right],
\end{align}
where the expectation is evaluated at the distribution $\PP_X$ of the covariates. 
A smaller absolute error suggests a more accurate HTE estimator, and a zero error implies $\hat{\tau}(X) = {\tau}(X)$ with probability one.
% Eq.~\eqref{defi:evaluation.error}  is a direct extension of the prediction error commonly used in regression.
The relative error of two estimators $\hat{\tau}_1(x)$ and $\hat{\tau}_2(x)$ is quantified as the difference between their absolute errors 
\begin{align}\label{defi:evaluation.error.relative}
    \delta(\hat{\tau}_1, \hat{\tau}_2)
    := \phi(\hat{\tau}_1(x)) - \phi(\hat{\tau}_2(x))  
    = \EE\left[\hat{\tau}_1^2(X) - \hat{\tau}_2^2(X) - 2(\hat{\tau}_1(X) - \hat{\tau}_2(X)) \tau(X) \right].
\end{align}  
A negative $\delta(\hat{\tau}_1, \hat{\tau}_2)$ indicates that $\hat{\tau}_1(x)$ is more accurate; otherwise, $\hat{\tau}_2(x)$ is more accurate.
In \Cref{rmk:DINA}, we consider the treatment effects defined on the natural parameter scale suitable for binary, count, and survival responses, and extend the results regarding Eq.~\eqref{defi:evaluation.error} to the corresponding errors.

For standard predictor evaluation, absolute prediction errors are more commonly displayed than relative prediction errors.
However, for HTE estimators, we have the perhaps surprising observation that the relative error can often be approximated more accurately.
% and associated with a narrower confidence interval than the absolute counterpart. 
Intuitively, this is because the relative error $ \delta(\hat{\tau}_1, \hat{\tau}_2)$ is linear in the unobserved $\tau(x)$, while the absolute error $\phi(\hat{\tau}(x))$ also depends on the second moment of $\tau(x)$, and estimating the first moment of $\tau(x)$ is relatively easier that of the second moment. 
We provide rigorous characterizations and empirical comparison of the observation in the following sections.

\subsection{Literature of HTE evaluation}\label{sec:literature}

In this paper, we consider comparing HTE estimators with statistical confidence through relative error evaluation, which is different from the mainstream literature focusing on evaluating absolute errors without uncertainty quantification. 
Nevertheless, we provide a brief overview of existing methods on assessing the absolute performance of an HTE estimator.
One simple and common approach targets the observed response but not the treatment effect, and uses the standard prediction error of the response as the error measurement.
However, an accurate predictor of the treatment effect may not necessarily be associated with precise response predictors \parencite{curth2023search}.
Furthermore, for HTE estimators that directly estimate the difference without estimating the response, such as causal trees \parencite{athey2016recursive}, the response prediction error can not be computed.
Another thread of assessment approaches involves creating ``virtual twins'' by matching treated and control units and use the response difference of a pair as a pseudo-observation of the treatment effect \parencite{rolling2014model}.
% zijun gao
However, matching is often computationally intensive \parencite{rosenbaum1989optimal}. Moreover, the complexity of matching algorithms makes it difficult to analyze statistically and perform downstream inference.
A third thread of methods estimate the HTE on the test dataset and compare it with the provided HTE estimators, which we refer to as plug-in estimators. This can be problematic because the evaluation error is affected by the error from the HTE estimator obtained from the test set, a nuisance function in this causal assessment task.
To reduce the impact of the error of the HTE estimator from the test data, bias correction methods based on influence functions have been developed \parencite{alaa2019validating}. 
The method is related to ours, but it does not address the uncertainty quantification of the estimated error.
Additionally, the influence function in \parencite{alaa2019validating} is different from our efficient proposal (comparison can be found in \Cref{sec:simulation}).

Our approach for drawing inference on relative error evaluation builds on the rapidly advancing body of work in semi-parametrics, particularly influence functions \parencite{van2000asymptotic, robins2008higher} (see \cite{kennedy2022semiparametric} for a review).  
In many causal problems, the causal quantity of interest is typically a scalar or low-dimensional, but the model contains infinite dimensional nuisance functions, making it a semi-parametric problem.
The influence function is a powerful tool for constructing the so-called one-step correction estimators that is more robust to the error of the estimators of nuisance components.
Given a specific function class that the true distribution belongs to, the estimator that attains the minimal asymptotic variance is known as the semi-parametrically efficient estimator, and the corresponding influence function is referred to as the efficient influence function. Despite the appealing statistical properties of efficient influence functions, the derivation of efficient influence functions is often case-specific, with only a few general rules and standard examples available \parencite{kennedy2022semiparametric}.

\section{Absolute error and issues}\label{sec:absolute.error}

In this section, we begin by discussing the estimation and inference of the absolute error using influence functions. We then highlight the undesirable properties of the absolute error estimation, which leads to our proposal of using relative error in \Cref{sec:relative.error}.

\subsection{Absolute error estimation via influence functions}\label{sec:absolute.error.estimation}
% Though XXX has developed a one-step correction estimator for $ \phi(\hat{\tau})$ in~\eqref{defi:evaluation.error} using the influence function in (XXX, Theorem 2), 
Provided with an HTE estimator $\hat{\tau}(x)$,
we adopt the following influence function for $\phi(\hat{\tau})$, 
\begin{align}\label{eq:EIF}
\begin{split}    
    \psi(\phi(\hat{\tau}); Z)
    :=&  \left((\mu_1(X) - \mu_0(X)) - \hat{\tau}(X)\right)^2 \\
    &+ 2\left((\mu_1(X) - \mu_0(X)) - \hat{\tau}(X)\right) \cdot \left(\frac{W(Y - \mu_1(X))}{e(X)} - \frac{(1-W)(Y - \mu_0(X))}{1-e(X)} \right) - \phi(\hat{\tau}).
\end{split}
\end{align}
Our derivation of the efficient influence function aligns with the variable importance measurement for heterogeneous treatment effects \parencite{hines2022variable}.
% We use the influence function for common causal quantities, such as the weighted average treatment effect, and apply the principles of influence function derivation for a linear combination of causal estimands and that for a function of a causal estimand.
The one-step correction estimator associated with \eqref{eq:EIF} is
\begin{align}\label{eq:estimator.absolute.error}
\begin{split}    
    \hat{\phi}(\hat{\tau})
    :=& \frac{1}{n} \sum_{i=1}^n \hat{\psi}(\phi(\hat{\tau}); Z_i)
    = \frac{1}{n} \sum_{i=1}^n \left((\tilde{\mu}_1(X_i) - \tilde{\mu}_0(X_i)) - \hat{\tau}(X_i)\right)^2  \\
    &+ 2\left((\tilde{\mu}_1(X_i) - \tilde{\mu}_0(X_i)) - \hat{\tau}(X_i)\right) \cdot \left(\frac{W_i(Y_i - \tilde{\mu}_1(X_i))}{\tilde{e}(X_i)} - \frac{(1-W_i)(Y_i - \tilde{\mu}_0(X_i))}{1-\tilde{e}(X_i)} \right).
\end{split}
\end{align}
Here $\tilde{\mu}_0(x)$, $\tilde{\mu}_1(x)$, and $\tilde{e}(x)$ are estimators of the nuisance functions ${\mu}_0(x)$, ${\mu}_1(x)$, and $e(x)$ obtained on the test dataset\footnote{We use cross-fitting \parencite{chernozhukov2018double} to ensure the independence of $\tilde{\mu}_0(x)$, $\tilde{\mu}_1(x)$, and $\tilde{e}(x)$ used by $\hat{\phi}_i(\hat{\tau})$ are independent of $Y_i$, $W_i$, $X_i$ therein when computing $\hat{\phi}(\hat{\tau})$.}.
The variance of the estimated evaluation error, denoted by ${V}(\hat{\phi}(\hat{\tau}))$, can be approximated by the empirical variance of $\hat{\psi}(\phi(\hat{\tau}); Z_i)$, denoted by $\hat{V}(\hat{\phi}(\hat{\tau}))$.
The $1-\alpha$ confidence interval of the absolute error takes the form 
\begin{align}\label{eq:CI.absolute.error}
    [\underline{\hat{\phi}}(\hat{\tau}; 1-\alpha), ~\overline{\hat{\phi}}(\hat{\tau}; 1-\alpha)]:= \left[\hat{\phi}(\hat{\tau}) - q_{1-\alpha/2} \hat{V}^{1/2}(\hat{\phi}(\hat{\tau})), ~\hat{\phi}(\hat{\tau}) + q_{1-\alpha/2} \hat{V}^{1/2}(\hat{\phi}(\hat{\tau}))\right],
\end{align} 
where $q_{1-\alpha/2}$ denotes the $1-\alpha/2$ quantile of a standard normal random variable.
The entire algorithm is summarized in \Cref{algo:absolute.error}.
% When provided with two HTE estimators $\hat{\tau}_1(x)$, $\hat{\tau}_2(x)$, we construct the two confidence intervals... 

The theorem below characterize the asymptotic distribution of $\hat{\phi}(\hat{\tau})$, which provides the theoretical guarantee of the validity of the confidence interval~\eqref{eq:CI.absolute.error}.
\begin{theorem}\label{theo:absolute.error}
    Assume the following conditions.
    \begin{itemize}
        \item [(a)] $Y$ is bounded, $\eta < e(X) < 1 - \eta$ for some $\eta > 0$.
        \item [(b)] The nuisance function estimators $\tilde{\mu}_{0}(x)$, $\tilde{\mu}_{1}(x)$, $\tilde{e}(x)$ obtained from the test data\footnote{When cross-fitting is used to compute the absolute error, we require (b) to hold for all $\tilde{\mu}_{0}^{-k}(x)$, $\tilde{\mu}_{1}^{-k}(x)$, $\tilde{e}^{-k}(x)$, $1 \le k \le K$.} satisfy $\EE[(\tilde{\mu}_{1}(X) - \mu_1(X))^2]^{1/2}$, $\EE[(\tilde{\mu}_{0}(X) - \mu_0(X))^2]^{1/2}$, $\EE[(\hat{e}(X) - e(X))^2]^{1/2} = o_p(n^{-1/4})$, $1 \le k \le K$. 
        \item [(c)] The true absolute error $\phi(\hat{\tau}) > 0$.
    \end{itemize}
    Then $\hat{\phi}(\hat{\tau})$, $\hat{V} (\hat{\phi}(\hat{\tau}))$ of \Cref{algo:absolute.error} satisfy
    \begin{align*}
        \frac{\hat{\phi}(\hat{\tau}) - {\phi}(\hat{\tau})}{\sqrt{n\hat{V}(\hat{\phi}(\hat{\tau}))}}
        \stackrel{d}{\to} \calN(0,1).
    \end{align*}
    In addition, the estimator $\hat{\phi}(\hat{\tau})$ is semi-parametrically efficient regarding the nonparametric model.
\end{theorem}

The proof is provided in the appendix.
% In \Cref{sec:relative.error}, we show the doubly robust property of the relative error described in is even more compelling.
We remark that our influence function and the estimator is different from that of \cite[Theorem 2]{alaa2019validating}.
Since our proposal $\hat{\phi}(\hat{\tau})$ is semi-parametrically efficient, the confidence interval is proved to be no wider than that in \cite{alaa2019validating} (see \Cref{sec:simulation} for numerical evidence).
% as it only requires the product of $|(\tilde{\mu}_{1}^{-k}(X) - \mu_1(X))(\hat{e}^{-k}(X) - e(X))|$ and $ |(\tilde{\mu}_{0}^{-k}(X) - \mu_0(X))(\hat{e}^{-k}(X) - e(X))| $ to be of order $ n^{-1/2}$.
% and is less sensitive to the error in $\tilde{\mu}_{0}^{-k}(X)$ and $\tilde{\mu}_{1}^{-k}(X)$.
% both $\tilde{\mu}_{0}^{-k}(X)$ and $\tilde{\mu}_{1}^{-k}(X)$ underfit, leading to $\tilde{\mu}_{1}^{-k}(X) - \tilde{\mu}_{0}^{-k}(X) \approx 0$. 
% This result in a significant downward bias in the absolute error estimator. 
% In contrast, the relative error estimator remains nearly unbiased since $\hat{e}^{-k}(X)$ is accurate and $(\tilde{\mu}_{0}^{-k}(X) - \mu_0(X))(\hat{e}^{-k}(X) - e(X))$ is small in magnitude.
According to \Cref{theo:absolute.error}, the convergence of $\hat{\phi}(\hat{\tau})$ at the parametric rate of $ n^{-1/2} $ only requires all nuisance function estimators to converge at a rate no slower than $ n^{-1/4}$, that is the locally doubly robust property \parencite{chernozhukov2018double}.
% meaning that $\hat{\phi}(\hat{\tau})$ is resilient to the errors of the nuisance function estimators around their true values.

Based on the $1-\alpha$ confidence interval of the absolute errors for two HTE estimators $\hat{\tau}_1$, $\hat{\tau}_2$, if the absolute error interval of $\hat{\tau}_1$ lies entirely to the right of that of $\hat{\tau}_2$, we can conclude with at least $1-2\alpha$ confidence that the estimation error of $\hat{\tau}_1$ is greater than that of $\hat{\tau}_2$, and therefore $\hat{\tau}_2$ should be selected. 
If the two intervals overlap, we are unable to confidently decide which estimator is more accurate.

\subsection{Issues of absolute error}\label{sec:absolute.error.issue}

Despite the results in \Cref{theo:absolute.error}, we highlight several undesirable aspects of the absolute error approach. In \Cref{sec:relative.error} below, we demonstrate how relative error estimation can address these issues.
\begin{itemize}
    \item [(i)] Sensitivity to errors of nuisance function estimators. 
    In \Cref{fig:inaccurate nuisance function estimator}, errors in $\tilde{\mu}_1(x)$, $\tilde{\mu}_0(x)$, $\tilde{e}(x)$ can introduce significant bias into the absolute error estimator, even resulting in negative estimates conflicting with the fact that error should always be non-negative. 
    % Rather than requiring $\tilde{\mu}_0(x)$, $\tilde{\mu}_1(x)$, and $\tilde{e}(x)$ to be estimated at the rate $n^{-1/4}$, respectively, it would be preferable if the condition can be relaxed to only the products of the errors $(\tilde{\mu}_w(x) - \mu_w(x)) (\tilde{e}(x) - e(x))$, for $w \in \{0,1\}$, are required to converge at the rate $n^{-1/2}$.
    
    \item [(ii)] Correlation across the estimated absolute errors of different HTE estimators.
    The estimated absolute error of different HTE estimators are correlated because they are based on the same validation data and share the complex nuisance function estimators.
    \Cref{theo:absolute.error} guides the construction of confidence intervals for each estimator separately, but does not directly address the correlation between the estimated absolute errors.

    \item [(iii)] Degenerate null. Condition (c) of \Cref{theo:absolute.error} indicates the asymptotic distribution may be invalid for the degenerate case $\PP_X(\hat{\tau}(X) = \tau(X)) = 1$. 
    One solution is analyzing the asymptotic distribution of the higher-order pathwise derivatives of $\phi(\hat{\tau})$ to derive the asymptotic distribution of \eqref{eq:estimator.absolute.error} when $\tau(x) = \hat{\tau}(x)$, which remains generally an open problem \parencite{hines2022variable, hudson2023nonparametric}.
    % Similar degeneracy issues have been encountered in the study of variable importance \parencite{hines2022variable}, where their estimand is expressed as the difference between two quantities and a sample splitting procedure can be employed albeit with some loss of efficiency. However, in our case, there does not appear to be a straightforward way to decompose the quantity in a similar manner.
    % As illustrated later in \Cref{sec:relative.error}, the degeneracy is less an issue for the relative error.

\end{itemize}

\begin{figure}[h]
        \centering
        \begin{minipage}{0.3\textwidth}
                \centering
                \includegraphics[clip, trim = 0cm 0cm 0cm 0cm, width = \textwidth]{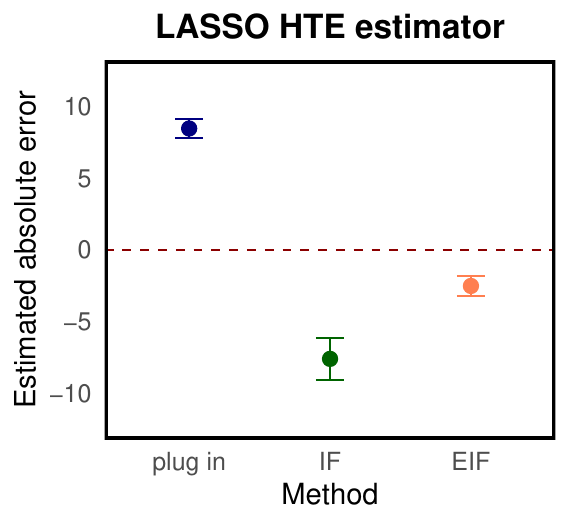}
        \end{minipage}
        \begin{minipage}{0.3\textwidth}
                \centering
                \includegraphics[clip, trim = 0cm 0cm 0cm 0cm, width = \textwidth]{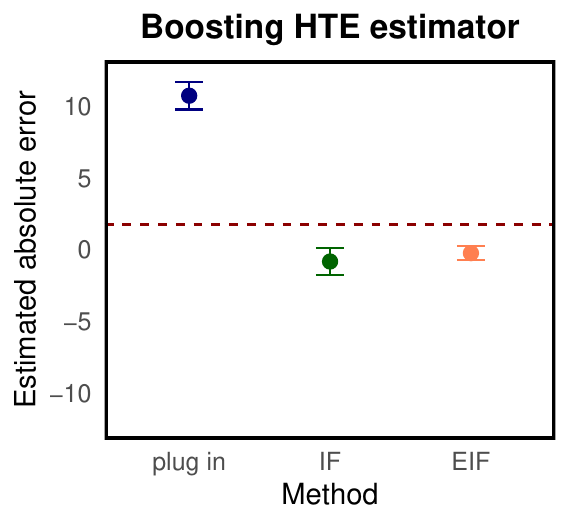}
        \end{minipage}
        \hspace{1cm}
            \begin{minipage}{0.3\textwidth}
                \centering
                \includegraphics[clip, trim = 0cm 0cm 0cm 0cm, width = \textwidth]{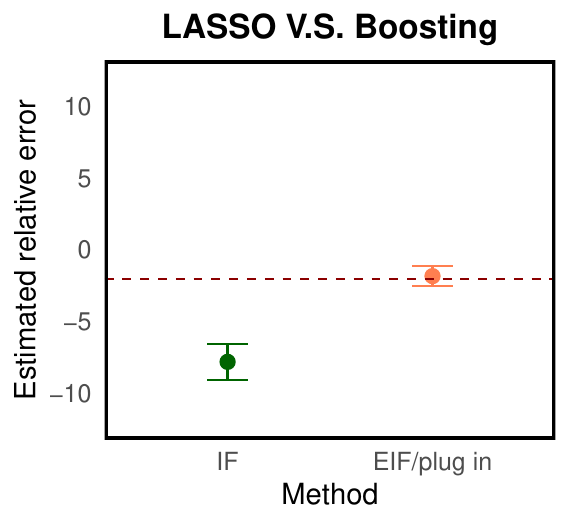}
        \end{minipage}
        \text{\hspace{3cm} (a) Absolute error \hspace{5cm} (b) Relative error}
        \caption{Comparison of estimated absolute and relative errors with inaccurate nuisance function estimators.
        We compare two HTE estimators: one uses LASSO for nuisance function estimation, and the other uses gradient boosting. 
        For both HTE estimators, we implement three approaches to evaluate the absolute and the relative errors: (1) the plug-in method, (2) the estimator by \cite{alaa2019validating} (IF), and (3) our proposal (EIF). 
        (Our proposal relative error estimator agrees with that based on the plug-in absolute error estimator, and thus we merge the two methods in the third panel.)
        In all methods, cross-fitting is used and the nuisance function estimators are shared. 
        All evaluation methods use the same nuisance function estimators derived from the test dataset, which were intentionally designed to underfit the data and thus inaccurate.
        % The propensity score is estimated by taking the mean of the treatment assignments. The outcome regression functions $ \mu_0(x) $ and $ \mu_1(x) $ are both estimated using gradient boosting to the outcomes of the control and treatment group data. 
        % However, to introduce deliberate inaccuracies in the nuisance functions, we limit the number of boosting rounds, leading to underfitting in the estimators for $ \mu_0(x) $ and $ \mu_1(x) $.
        For each evaluation, we plot the estimated error as well as the $90\%$ confidence interval.
        The red dashed line in the figure represents the true absolute/relative error.
        The IF and EIF estimators for absolute error are negative, which conflicts with the non-negative nature of prediction errors. 
        Only the confidence interval of the EIF relative error correctly captures the true value.
        }
    \label{fig:inaccurate nuisance function estimator}
\end{figure}

\section{Relative error estimation and inference}\label{sec:relative.error}

In this section, we focus on the estimation and inference of relative error based on influence functions. We then compare the relative error estimation with the absolute counterpart regarding the issues in \Cref{sec:absolute.error.issue}.

\subsection{Relative error estimation via influence functions}\label{sec:relative.error.estimation}

Provided with two HTE estimators $\hat{\tau}_1(x)$, $\hat{\tau}_2(x)$, we adopt the following influence function for the relative error $\delta(\hat{\tau}_1, \hat{\tau}_2)$,
% Based on \Cref{eq:EIF} and the principle that the influence function of two causal quantities is the difference of the their influence functions , we arrive at
\begin{align}\label{eq:EIF.relative} 
\begin{split}    
    &\psi(\delta(\hat{\tau}_1, \hat{\tau}_2); Z)
    := \hat{\tau}_1^2(X) - \hat{\tau}_2^2(X) \\
    &- 2\left(\hat{\tau}_1(X) - \hat{\tau}_2(X)\right) \cdot \left(\frac{W(Y - \mu_1(X))}{e(X)} + \mu_1(X)- \frac{(1-W)(Y - \mu_0(X))}{1-e(X)} - \mu_0(X)\right) - \delta(\hat{\tau}_1, \hat{\tau}_2).
\end{split}
\end{align}
The implied one-step correction estimator is
\begin{align}\label{eq:estimator.relative.error}
\begin{split}    
    \hat{\delta}(\hat{\tau}_1, \hat{\tau}_2)
    :=&\frac{1}{n} \sum_{i=1}^n  \hat{\psi}(\delta(\hat{\tau}_1, \hat{\tau}_2); Z_i)
    = \frac{1}{n} \sum_{i=1}^n \hat{\tau}_1^2(X_i) - \hat{\tau}_2^2(X_i) \\
    &- 2\left(\hat{\tau}_1(X_i) - \hat{\tau}_2(X_i)\right) \cdot \left(\frac{W_i(Y_i - \tilde{\mu}_1(X_i))}{\tilde{e}(X_i)} + \tilde{\mu}_1(X_i)- \frac{(1-W_i)(Y_i - \tilde{\mu}_0(X_i))}{1-\tilde{e}(X_i)} - \tilde{\mu}_0(X_i)\right).
\end{split}
\end{align}
The algorithm can be summarized as \Cref{algo:absolute.error} combined with \eqref{eq:estimator.relative.error}.

Based on the $1-\alpha$ confidence interval of the relative error, if the interval lies entirely to the right of zero, we can conclude with at least $1-\alpha$ confidence that the estimation error of $\hat{\tau}_1$ is greater than that of $\hat{\tau}_2$, and therefore $\hat{\tau}_2$ should be selected. 
If the interval lies entirely to the left of zero, we can conclude with the same confidence that the estimation error of $\hat{\tau}_2$ is greater than that of $\hat{\tau}_1$, and $\hat{\tau}_1$ should be chosen. 
If the interval contains zero, we are unable to confidently decide which estimator is superior.

\begin{algorithm}\caption{Absolute (relative) error}\label{algo:absolute.error}
    \begin{algorithmic}[1]
    \STATE \textbf{Input}: An HTE estimator $\hat{\tau}(x)$, test data $Z_i = (X_i, W_i, Y_i)$, $1 \le i \le n$, methods of estimating nuisance functions $\mu_0(x)$, $\mu_1(x)$, $e(x)$, number of folds $K$ for cross-fitting, confidence level $1-\alpha$.
    \STATE Randomly split the test dataset into $K$ folds of approximately equal size. Denote the $k$-th fold by $D_k$.

    \FOR{$k = 1,\ldots,K$}
    
    \STATE Apply the nuisance function estimators to test folds $\cup_{j \neq k} D_j$ and obtain $\tilde{\mu}_{0}^{-k}(x)$, $\tilde{\mu}_{1}^{-k}(x)$, and $\tilde{e}^{-k}(x)$.

    \ENDFOR
    
    \STATE Compute the one-step correction estimator based on \Cref{eq:estimator.absolute.error} (\Cref{eq:estimator.relative.error}). 
    For $i \in D_k$,
    \begin{align*} 
        \hat{\psi}_i^+(\phi(\hat{\tau}); Z_i) :=& \left((\tilde{\mu}_1^{-k}(X_i) - \tilde{\mu}_0^{-k}(X_i)) - \hat{\tau}(X_i)\right)^2 \\
        &+ 2\left((\tilde{\mu}_1^{-k}(X_i) - \tilde{\mu}_0^{-k}(X_i)) - \hat{\tau}(X_i)\right) \cdot \left(\frac{W_i(Y_i - \tilde{\mu}_1^{-k}(X_i))}{\tilde{e}^{-k}(X_i)} - \frac{(1-W_i)(Y_i - \tilde{\mu}_0^{-k}(X_i))}{1-\tilde{e}^{-k}(X_i)} \right),
    \end{align*}
     and take the average
     \begin{align*}
        \hat{\phi}(\hat{\tau}) :=& \frac{1}{n} \sum_{k=1}^K \sum_{i \in D_k} \hat{\psi}_i^+(\phi(\hat{\tau}); Z_i).
    \end{align*}
    Here $\hat{\psi}_i^+(\phi(\hat{\tau}); Z_i) $ and $\hat{\psi}_i(\phi(\hat{\tau}); Z_i)$ differ in a constant.
    
    \STATE Compute the estimator of the variance of $\hat{\phi}(\hat{\tau})$,
    \begin{align*}
       \hat{V}(\hat{\phi}(\hat{\tau})) :=& \frac{1}{n} \sum_{i = 1}^n \left(\hat{\psi}_i^+(\phi(\hat{\tau}); Z_i)  -  \hat{\phi}(\hat{\tau}) \right)^2.
    \end{align*}
    
    \STATE \textbf{Output}: Estimated error $\hat{\phi}(\hat{\tau})$, the estimator of its variance $\hat{V}(\hat{\phi}(\hat{\tau}))$, $1-\alpha$ confidence interval $\left[\hat{\phi}(\hat{\tau}) - q_{1-\alpha/2} \hat{V}^{1/2}(\hat{\phi}(\hat{\tau})), ~\hat{\phi}(\hat{\tau}) + q_{1-\alpha/2} \hat{V}^{1/2}(\hat{\phi}(\hat{\tau}))\right]$.

    \end{algorithmic}
\end{algorithm}

Similar to \Cref{theo:absolute.error}, we characterize the asymptotic distribution of $\hat{\delta}(\hat{\tau}_1, \hat{\tau}_2)$.
The proof is provided in the appendix.
\begin{theorem}\label{theo:relative.error}
    Assume the following conditions.
    \begin{itemize}
        \item [(a)] $Y$ is bounded, $\eta< e(X) < 1 - \eta$ for some $\eta > 0$.
        \item [(b)] The nuisance function estimators obtained from the test data satisfy $\|\tilde{\mu}_{1}(X) - \mu_1(X)\|_2$, $\|\tilde{\mu}_{0}(X) - \mu_0(X)\|_2$, $\|\hat{e}(X) - e(X)\|_2 = o_p(1)$, and $\EE[|(\tilde{\mu}_{1}(X) - \mu_1(X))(\hat{e}(X) - e(X))|]$, $\EE[|(\tilde{\mu}_{0}(X) - \mu_0(X))(\hat{e}(X) - e(X))|] = o_p(n^{-1/2})$.
        \item [(c)] The true relative error $\EE[(\hat{\tau}_1(X) - \hat{\tau}_2(X))^2] \neq 0$.
    \end{itemize}
    Then $\hat{\delta}(\hat{\tau}_1, \hat{\tau}_2)$, $\hat{V}(\hat{\delta}(\hat{\tau}_1, \hat{\tau}_2))$ of \Cref{algo:absolute.error} satisfy
    \begin{align*}
        \frac{ \hat{\delta}(\hat{\tau}_1, \hat{\tau}_2) -     {\delta}(\hat{\tau}_1, \hat{\tau}_2)}{\sqrt{n\hat{V}(    \hat{\delta}(\hat{\tau}_1, \hat{\tau}_2))}}
        \stackrel{d}{\to} \calN(0,1).
    \end{align*}
    The estimator $\hat{\delta}(\hat{\tau}_1, \hat{\tau}_2)$ is efficient regarding the nonparametric model.

    Further assume
      \begin{itemize}
        \item [(d)] $\hat{e}(x) = e(x)$ or $\tilde{\mu}_{1}(x) = \mu_1(x)$, $\tilde{\mu}_{0}(x) = \mu_0(x)$.
    \end{itemize}
    Then the estimator $\hat{\delta}(\hat{\tau}_1, \hat{\tau}_2)$ is unbiased, i.e., $\EE\left[\hat{\delta}(\hat{\tau}_1, \hat{\tau}_2)\right]
        = {\delta}(\hat{\tau}_1, \hat{\tau}_2)$.
\end{theorem}

\subsection{Advantages of relative error}\label{sec:relative.error.advantage}

We compare \Cref{theo:relative.error} and \Cref{theo:absolute.error}, and discuss how the estimation and inference of relative errors avoids the issues of absolute errors in \Cref{sec:absolute.error.issue}.

\begin{itemize}
    \item [(i)] Condition (b) of \Cref{theo:relative.error} only requires the nuisance function estimators to be consistent and the product of the error of $\hat{e}(x)$ and $\tilde{\mu}_{1}(x)$, $\tilde{\mu}_{0}(x)$ to converge at $n^{-1/2}$.
    This condition of relative error is strictly weaker than the condition (b) of the absolute error in \cref{theo:absolute.error}.
    % \Cref{fig:inaccurate nuisance function estimator} demonstrates numerically that the  \cref{theo:absolute.error} is more susceptible to inaccurate nuisance function estimators.
    Particularly, when the treatment assignment is known, such as in randomized trials, then the validity of \Cref{theo:relative.error} only requires the consistency of $\tilde{\mu}_{1}(x)$, $\tilde{\mu}_{0}(x)$.

    In addition, condition (d) implies that the relative error estimator satisfies a global doubly robust property, meaning that if $\hat{e}(x) = e(x)$, then $\hat{\mu}_1(x)$ and $\hat{\mu}_0(x)$ can be arbitrary, even inconsistent, and the estimated relative error will still be unbiased. Similarly, the relative error remains unbiased if $\hat{\mu}_1(x)$ and $\hat{\mu}_0(x)$ are correct, regardless of the estimator $\hat{e}(x)$ used.
    This property does not hold for the absolute error estimator.

    \item [(ii)] Relative error estimators directly estimates and performs inference for the relative difference between two models' performance, rather than dealing with each model's error separately and then comparing them. 
    Therefore, unlike the absolute error estimation approach, there is no need to handle the dependency between two absolute error confidence intervals.
    
    \item [(iii)] Degenerate null. 
    Condition (c) of \Cref{theo:absolute.error} $\hat{\tau}(x) \neq \tau(x)$ is hard to validate or falsify, since $\tau(x)$ is not directly observed and needs to be estimated.
    In contrast, condition (c) of \Cref{theo:relative.error} $\hat{\tau}_1(x) \neq \hat{\tau}_2(x)$ can be verified straightfowardly, since $\hat{\tau}_1(x)$ and $\hat{\tau}_2(x)$ are provided functions and there is no extra estimation required.

    \item [(iv)] The relative error is more effective in identifying the better HTE estimators compared to the absolute error when $\hat{\tau}_1$, $\hat{\tau}_2$ are similar (see \Cref{fig:similar HTE estimator} for a numerical example).
    % To determine the better estimator based on the confidence intervals of absolute errors, we can conclude $\hat{\tau}_1$ is better if $\underline{\hat{\phi}}(\hat{\tau}_1; 1 - \alpha/2) > \overline{\hat{\phi}}(\hat{\tau}_2; 1 - \alpha/2)$.
    When $\hat{\tau}_1(X)$ and $\hat{\tau}_2(X)$ are similar, the confidence intervals for the absolute errors of $\hat{\tau}_1(X)$ and $\hat{\tau}_2(X)$ can still be wide, and the absolute error confidence intervals are too wide to be informative of which estimator is better.
    In contrast, the width of the confidence interval for $\hat{\delta}(\hat{\tau}_1, \hat{\tau}_2)$ is proportional to $\sqrt{\EE[(\hat{\tau}_1(X) - \hat{\tau}_2(X))^2]}$, which we prove in the appendix, and the confidence interval of the relative error can still be sufficiently narrow to reliably identify the more accurate estimator.   
\end{itemize}

\begin{figure}[h]
        \centering
        \begin{minipage}{0.5\textwidth}
                \centering
                \includegraphics[clip, trim = 0cm 0cm 0cm 0cm, width = \textwidth]{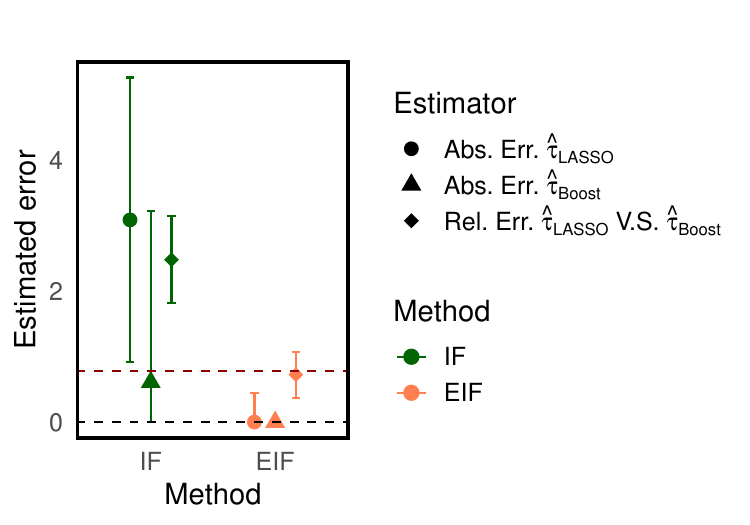}
        \end{minipage}
        \caption{
        Relative error evaluation estimators are more effective in determining the better one of two similar HTE estimators.
        We consider two similar HTE estimators using LASSO for nuisance function estimation which only differ in the regularization hyperparameter. 
        We consider the IF \parencite{alaa2019validating}, EIF (our proposal) absolute, relative error confidence intervals.
        The $90\%$ absolute error confidence intervals of both IF and EIF are too wide to distinguish the two HTE estimators, while the confidence interval of the relative error estimator is significantly narrower and find the better estimator $\hat{\tau}_2$ (the true relative error is indicated in dark red and $\hat{\tau}_2$ has a smaller error).
        We remark that even if IF's relative error confidence interval stays above zero, it does not contain the true relative error.
        }
    \label{fig:similar HTE estimator}
\end{figure}

\subsection{Beyond difference in conditional means}\label{rmk:DINA}

For responses generated from an exponential family or the Cox model, such as when $ Y \mid X $ follows a Poisson distribution, the difference in the natural parameter functions between the treatment and the control has been proposed as the treatment effect estimand \parencite{gao2022DINA}.
The difference in natural parameter functions is better suited for modeling covariate dependence, as transforming to the natural parameter scale eliminates support constraints.
Let $\eta(\cdot)$ be the link function transforming the conditional mean function to the natural parameter scale, then the DINA estimand is 
\begin{align*}
    \tau(x) := \eta(E[Y(1) \mid X = x]) - \eta(E[Y(0) \mid X = x]).
\end{align*} 
Similarly we can define the absolute error~\eqref{defi:evaluation.error}, the relative error~\eqref{defi:evaluation.error.relative}, 
\begin{align*}
    &\hat{\phi}(\hat{\tau})
    := \frac{1}{n} \sum_{i=1}^n \left((\tilde{\eta}_1(X_i) - \tilde{\eta}_0(X_i)) - \hat{\tau}(X_i)\right)^2  \\
    &+ 2\left((\tilde{\eta}_1(X_i) - \tilde{\eta}_0(X_i)) - \hat{\tau}(X_i)\right) \cdot \left(\frac{W_i\tilde{\eta}'_1(X_i)(Y_i - \tilde{\mu}_1(X_i))}{\tilde{e}(X_i)} - \frac{(1-W_i)\tilde{\eta}'_0(X_i)(Y_i - \tilde{\mu}_0(X_i))}{1-\tilde{e}(X_i)} \right),
\end{align*}
and derive the efficient one-step correction estimators, \begin{align*}    
    &\hat{\phi}(\delta(\hat{\tau}_1, \hat{\tau}_2))
    := \frac{1}{n} \sum_{i=1}^n \hat{\tau}_1^2(X_i) - \hat{\tau}_2^2(X_i) - 2\left(\hat{\tau}_1(X_i) - \hat{\tau}_2(X_i)\right)\\
    & \cdot \left(\frac{W_i \tilde{\eta}'_1(X_i) (Y_i - \tilde{\mu}_1(X_i))}{\tilde{e}(X_i)} + \tilde{\eta}_1(X_i)- \frac{(1-W_i)\tilde{\eta}'_0(X_i)(Y_i - \tilde{\mu}_0(X_i))}{1-\tilde{e}(X_i)} - \tilde{\eta}_0(X_i)\right).
\end{align*}
Here $\tilde{\mu}_1(x)$, $\tilde{\mu}_0(x)$ are estimated nuisance mean functions, and $\tilde{\eta}_1(x)$, $\tilde{\eta}_0(x)$ are estimated natural parameter functions.
The estimators recover \eqref{eq:estimator.absolute.error} and \eqref{eq:estimator.relative.error} when the link function $\eta(\cdot)$ is identity. 
In the appendix, we present analogous results to \Cref{theo:absolute.error} and \Cref{theo:relative.error}.
The relative error of DINA maintains the advantages in \Cref{sec:relative.error.advantage}.

\section{Hypothetical real data analysis}\label{sec:simulation}

The ACIC 2016 competition dataset \parencite{dorie2019automated}, derived from the real-world Collaborative Perinatal Project \parencite{niswander1972women}, is used as a benchmark dataset for HTE estimation and related tasks. 
The dataset includes $55$ real variables of various types with natural associations. 
Treatment assignments and potential outcomes are generated using a variety of models. 
Therefore, the true $\mu_0(x)$, $\mu_1(x)$ and thus true $\tau(x)$ are known.
We select four representative scenarios varying two key factors: the model of the propensity score and the model of the group conditional functions. 
Specifically, the scenarios include:
\begin{itemize}
    \item [(a)] Linear $\mu_0(x)$, $\mu_1(x)$, linear $e(x)$.
    \item [(b)] Nonlinear $\mu_0(x)$, $\mu_1(x)$, linear $e(x)$.
    \item [(c)] Linear $\mu_0(x)$, $\mu_1(x)$, nonlinear $e(x)$. 
    \item [(d)] Nonlinear $\mu_0(x)$, $\mu_1(x)$, nonlinear $e(x)$.
\end{itemize} 
Here, linear functions refer to a linear combination of covariates $x$, while nonlinear functions incorporate quadratic or higher-order terms of the covariates. 
For both linear and non-linear functions, the nuisance functions only depend on less than $20\%$ of the covariates.

The total sample size is $4802$. We use a randomly sampled $2802$ data points to obtain the HTE estimators for comparison, and the rest $2000$ data points to evaluate and compare the estimators' performance.

In terms of the HTE estimators for comparison, we consider two T-learners \parencite{Kunzel4156}.
The estimators first estimate $\mu_0(x)$, $\mu_1(x)$ separately and then take the difference $\hat{\mu}_1(x) - \hat{\mu}_0(x)$.
The key difference between two HTE estimators lies in the machine learning algorithms used to obtain $\hat{\mu}_0(x)$, $\hat{\mu}_1(x)$. 
\begin{itemize}
    \item [(i)] LASSO. This method uses LASSO\footnote{We use the R package \textsc{glmnet} for implementation. In particular, we use \textsc{cv.glmnet} to choose the regularization parameter.} to estimate $\mu_0(x)$ and $\mu_1(x)$.
    \item [(ii)] Boosting. This method uses gradient boosting\footnote{We use \textsc{xgboost} for implementation.} to estimate $\mu_0(x)$ and $\mu_1(x)$.
\end{itemize}
For settings (a) and (c), $\mu_0(x)$, $\mu_1(x)$ are linear and LASSO is well-specified. Therefore, the LASSO-based T-learner is expected to be more accurate. 
In contrast, for settings (b) and (d), $\mu_0(x)$, $\mu_1(x)$ are nonlinear, and the boosting-based T-learner should perform more favorably.

We implement five assessment methods: three based on absolute error estimation and two based on relative error estimation.
For absolute error estimators, we consider
\begin{itemize}
    \item One-step correction estimator based on our efficient influence function (EIF), as defined in \eqref{eq:estimator.absolute.error}.
    \item One-step correction estimator by \cite{alaa2019validating} (IF).
    \item Plug-in estimator (plug in): $\sum_{i=1}^n (\hat{\tau}(X_i) - \tilde{\tau}(X_i))^2 / n$, where $\tilde{\tau}(x)$ is the AIPW estimator of the HTE $\tau(x)$ obtained from the test dataset.
\end{itemize}
For relative error estimators, we consider
\begin{itemize}
    \item One-step correction estimator based on our efficient influence function (EIF) in \eqref{eq:estimator.relative.error}. 
    We note that the relative error estimator implied by the plug-in estimator of the absolute error is equivalent to EIF relative error estimator.
    Therefore, we merge the two methods in display.
    
    \item One-step correction estimator of the relative error based on the influence function in \cite{alaa2019validating}.
\end{itemize}
We use 2-fold cross-fitting.
All five methods share the same nuisance function estimators for fair comparison, which we specify case by case below.

To compute the coverage and average width of the confidence intervals ($90\%$ confidence level), and the probability of selecting the correct model (selection accuracy), for scenarios (a) through (d), we generate $100$ different realizations of $\mu_0(x)$, $\mu_1(x)$, and $e(x)$, adhering to the linear/nonlinear constraints aforementioned. For each realization of $\mu_0(x)$, $\mu_1(x)$, and $e(x)$, we further simulate $Y$ and $Z$ independently $100$ times to compute the frequency that the confidence intervals cover the true errors.
Similarly for other metrics above.

\subsection{Confidence interval coverage}\label{sec:coverage}

In \Cref{fig:simulation.coverage}, we examine the simplest scenario (a), where both the propensity score function $e(x)$ and the outcome functions $\mu_0(x)$ and $\mu_1(x)$ are linear. We compare the five evaluation methods equipped with three nuisance functions:
\begin{itemize}
    \item [(i)] The true nuisance functions.
    \item [(ii)] A well-specified linear model.
    \item [(iii)] An erroneous gradient boosting learner.
\end{itemize}

For relative errors, the EIF confidence interval is consistently valid, robust to the nuisance function estimator provided.
In contrast, the coverage of the relative error IF confidence interval fails to achieve the correct coverage.
For absolute errors, the coverage depends heavily on the accuracy of the nuisance learners. 
When the true nuisance estimators are used, the coverage of the EIF method for both the LASSO HTE estimator and the Boosting HTE estimator is close to the target level. However, as we shift to the nuisance estimator obtained using linear regression (which is well-specified but not as accurate), the coverage decreases, and it almost never covers the true value when the inaccurate gradient boosting with underfitting is used. 
The absolute error estimator IF consistently undercovers, even with true nuisance functions.
Additionally, as shown in the appendix, 
% \Cref{fig:simulation.width}
the width of the confidence intervals based on IF, for both relative and absolute error estimators, is significantly larger than those based on EIF, indicating less precision.

\subsection{Probability of selecting the winner}\label{sec:selection.accuracy}

We present the probability of selecting the correct model in \Cref{fig:ACIC.selection.accuracy}. The analysis covers settings (a) through (d), where in (a) and (c), the LASSO HTE estimator outperforms the Boosting HTE estimator by a moderate margin, whereas in (b) and (d), the Boosting HTE estimator significantly outperforms the LASSO HTE estimator.
For absolute error confidence intervals, we select the better estimator if the two confidence intervals at the $\alpha/2$ level do not overlap. Otherwise, no selection is made. 
In contrast, for relative error confidence intervals, we select the better estimator if the confidence interval for the relative error does not contain zero.

For the relative error-based methods, we observe that across the four scenarios, the EIF relative error method achieves the highest selection accuracy, followed by the IF relative error method. 
For the absolute error-based methods, the selection accuracy of the EIF absolute error method is lower in settings (a) and (c). This is because the improvement of the LASSO HTE estimator over the Boosting counterpart is relatively modest, and the gap between them is overshadowed by the wide confidence intervals, resulting in the methods being unable to confidently make a selection. 
In settings (b) and (d), the gap between the two HTE estimators is more pronounced, leading to a higher probability of correct selection. 
for the plug-in and IF absolute error methods, the confidence intervals are too wide, preventing any confident conclusions made across scenarios.

\begin{figure}[h]
        \centering
        \begin{minipage}{0.3\textwidth}
                \centering
                \includegraphics[clip, trim = 0cm 0cm 0cm 0cm, width = \textwidth]{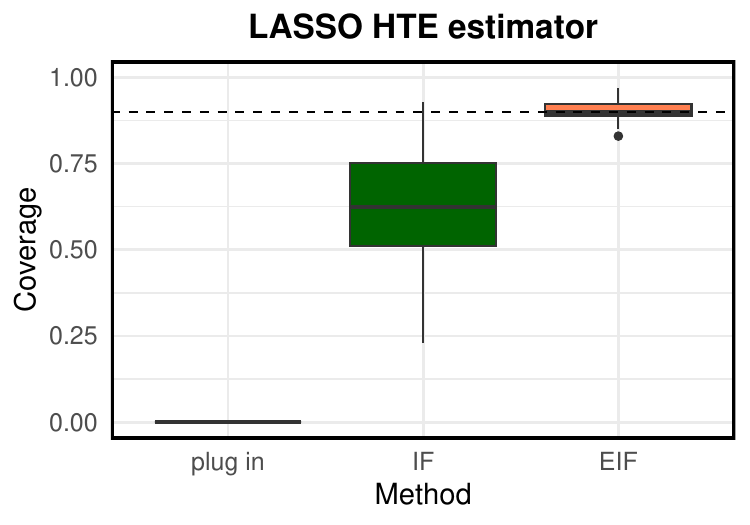}
                \subcaption*{}     
        \end{minipage}
        \begin{minipage}{0.3\textwidth}
                \centering
                \includegraphics[clip, trim = 0cm 0cm 0cm 0cm, width = \textwidth]{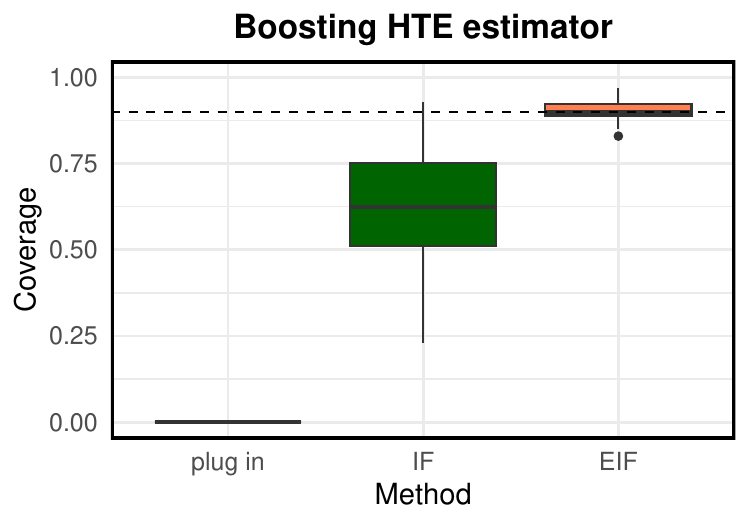}
                \subcaption*{(i) True}        
        \end{minipage}
        \begin{minipage}{0.3\textwidth}
                \centering
                \includegraphics[clip, trim = 0cm 0cm 0cm 0cm, width = \textwidth]{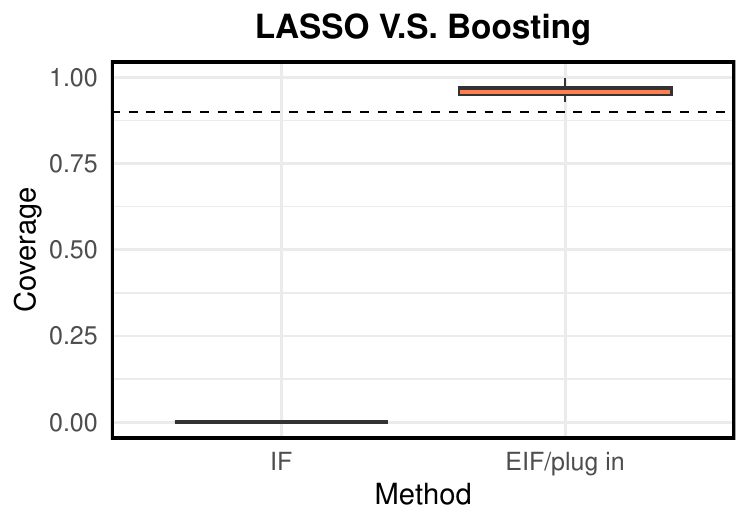}
                \subcaption*{}     
        \end{minipage}       
        \\ 
        \begin{minipage}{0.3\textwidth}
                \centering
                \includegraphics[clip, trim = 0cm 0cm 0cm 0cm, width = \textwidth]{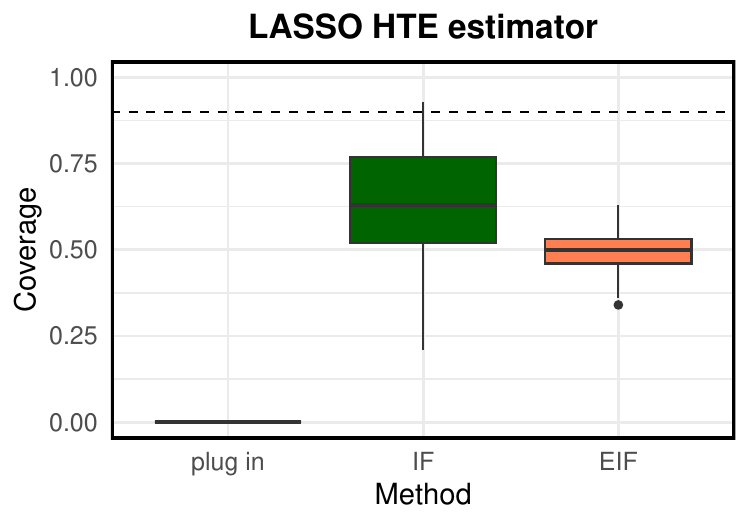}
                \subcaption*{}     
        \end{minipage}
        \begin{minipage}{0.3\textwidth}
                \centering
                \includegraphics[clip, trim = 0cm 0cm 0cm 0cm, width = \textwidth]{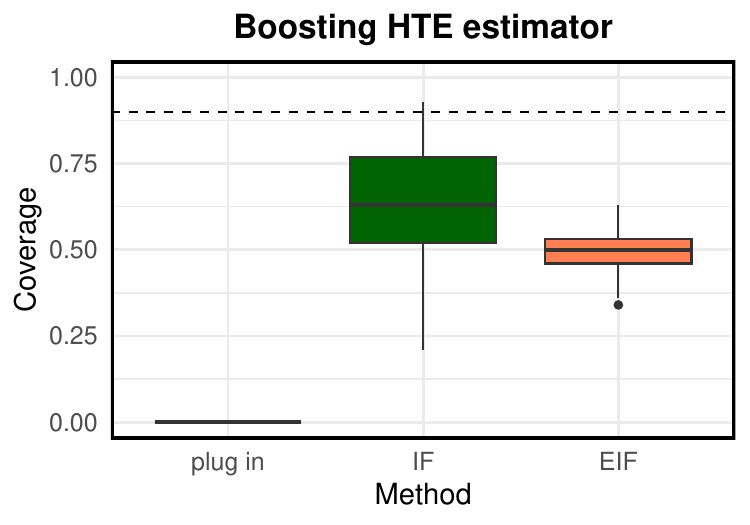}
                \subcaption*{(ii) Linear regression} 
        \end{minipage}
        \begin{minipage}{0.3\textwidth}
                \centering
                \includegraphics[clip, trim = 0cm 0cm 0cm 0cm, width = \textwidth]{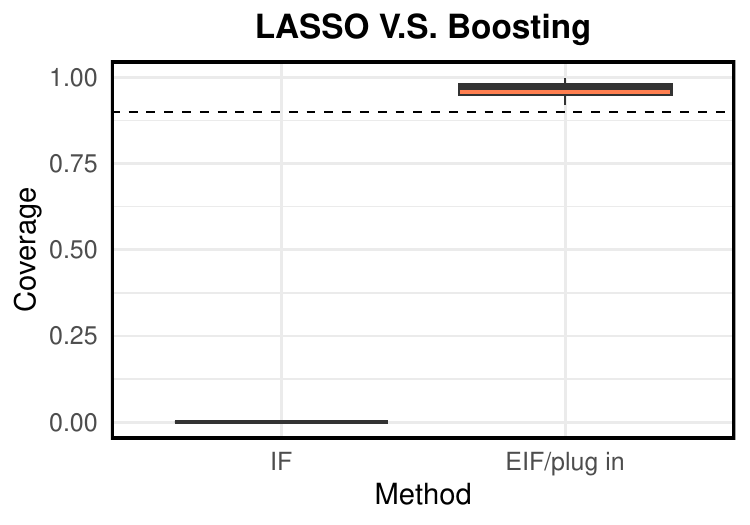}                \subcaption*{}
        \end{minipage}        \\
        \begin{minipage}{0.3\textwidth}
                \centering
                \includegraphics[clip, trim = 0cm 0cm 0cm 0cm, width = \textwidth]{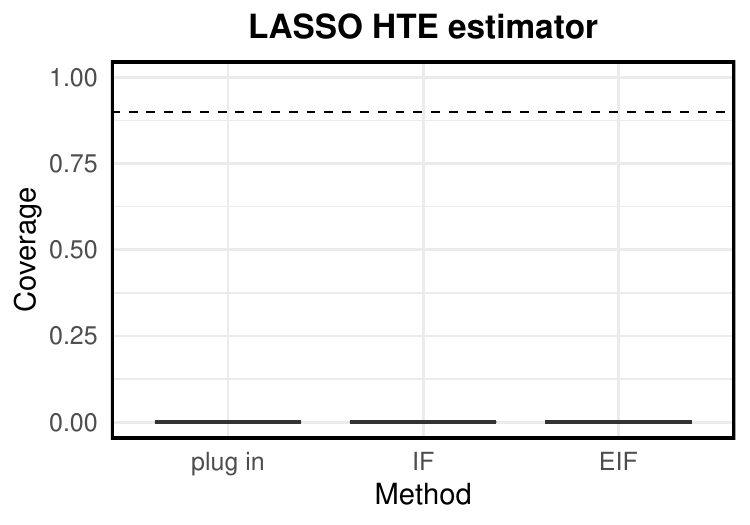}
                \subcaption*{}     
        \end{minipage}
        \begin{minipage}{0.3\textwidth}
                \centering
                \includegraphics[clip, trim = 0cm 0cm 0cm 0cm, width = \textwidth]{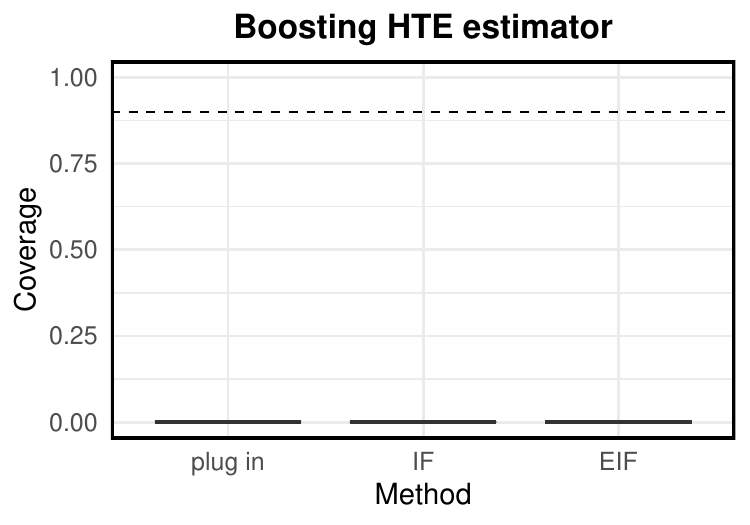}
                \subcaption*{(iii) Boosting with underfitting}
        \end{minipage}
        \begin{minipage}{0.3\textwidth}
                \centering
                \includegraphics[clip, trim = 0cm 0cm 0cm 0cm, width = \textwidth]{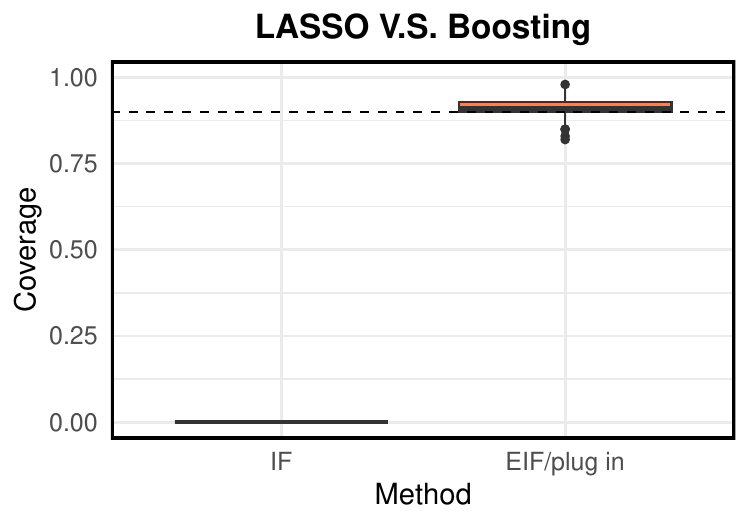}
                \subcaption*{}     
        \end{minipage}      
        \caption{ 
        Coverage of the estimated absolute (LASSO, Boosting)/relative (LASSO V.S. Boosting) error's $90\%$ confidence intervals across three methods (plug in, IF, EIF) over three nuisance functions ((i) true, (ii) estimated by linear regression, (iii) estimated by gradient boosting with underfitting).
        }
    \label{fig:simulation.coverage}
\end{figure}

\begin{figure}[h]
        \centering
        \begin{minipage}{0.3\textwidth}
                \centering
                \includegraphics[clip, trim = 0cm 0cm 0cm 0cm, width = \textwidth]{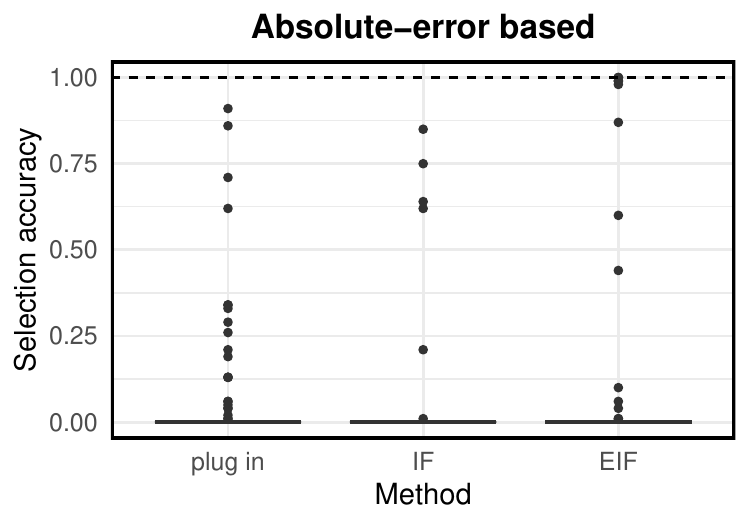}
        \end{minipage}
            \begin{minipage}{0.3\textwidth}
                \centering
                \includegraphics[clip, trim = 0cm 0cm 0cm 0cm, width = \textwidth]{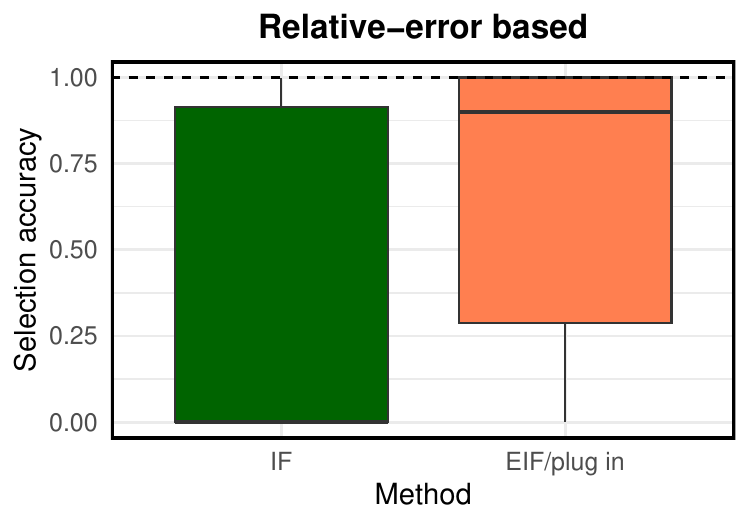}
        \end{minipage}
        \begin{center}            
        \text{(a) Linear $\mu_0(x)$, $\mu_1(x)$, linear $e(x)$} 
        \end{center}
                \begin{minipage}{0.3\textwidth}
                \centering
                \includegraphics[clip, trim = 0cm 0cm 0cm 0cm, width = \textwidth]{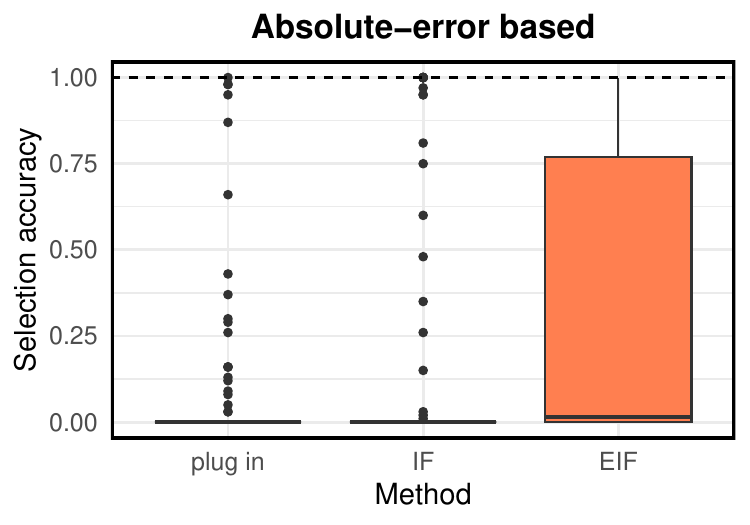}
        \end{minipage}
            \begin{minipage}{0.3\textwidth}
                \centering
                \includegraphics[clip, trim = 0cm 0cm 0cm 0cm, width = \textwidth]{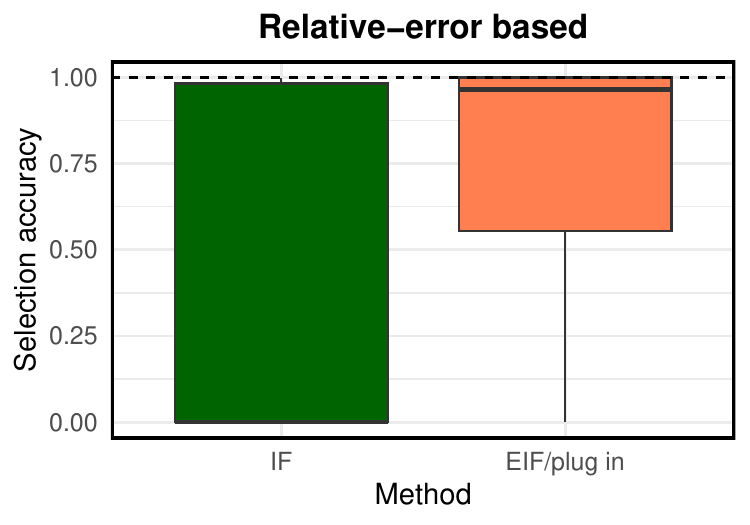}
        \end{minipage}
        \begin{center}            
        \text{(b) Nonlinear $\mu_0(x)$, $\mu_1(x)$, linear $e(x)$} 
        \end{center}
                \begin{minipage}{0.3\textwidth}
                \centering
                \includegraphics[clip, trim = 0cm 0cm 0cm 0cm, width = \textwidth]{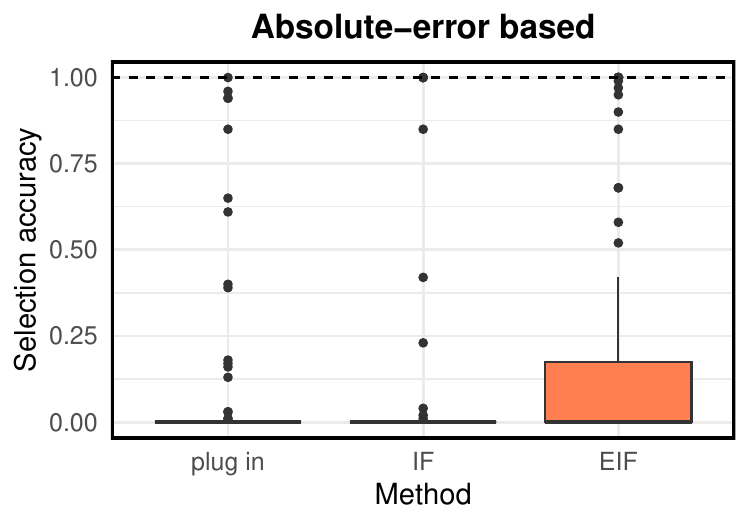}
        \end{minipage}
            \begin{minipage}{0.3\textwidth}
                \centering
                \includegraphics[clip, trim = 0cm 0cm 0cm 0cm, width = \textwidth]{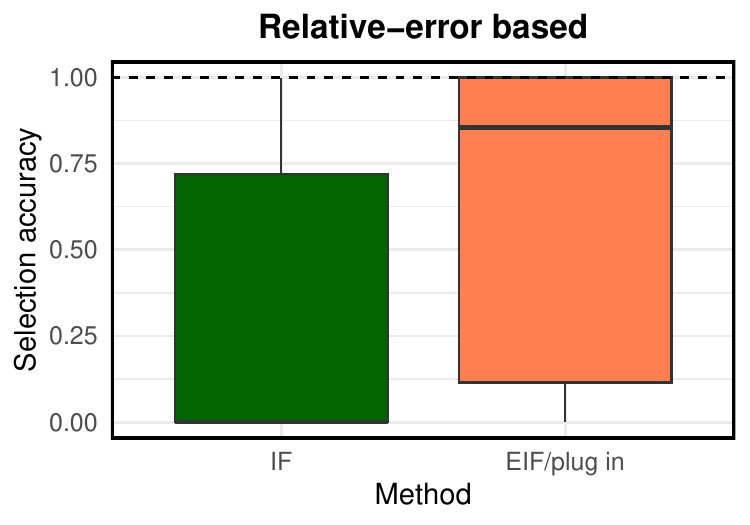}
        \end{minipage}
        \begin{center}            
        \text{(c) Linear $\mu_0(x)$, $\mu_1(x)$, nonlinear $e(x)$} 
        \end{center}
                \begin{minipage}{0.3\textwidth}
                \centering
                \includegraphics[clip, trim = 0cm 0cm 0cm 0cm, width = \textwidth]{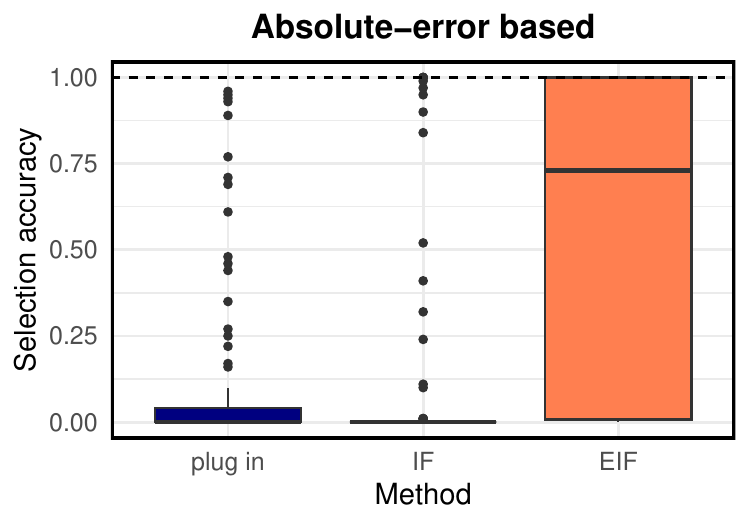}
        \end{minipage}
            \begin{minipage}{0.3\textwidth}
                \centering
                \includegraphics[clip, trim = 0cm 0cm 0cm 0cm, width = \textwidth]{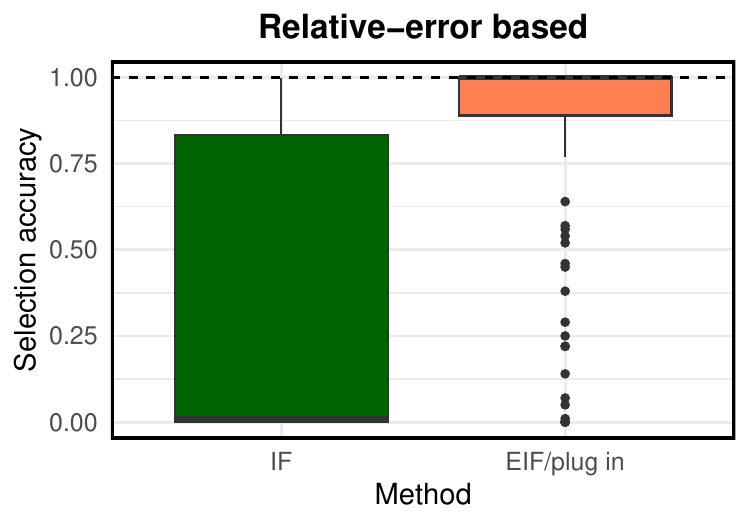}
        \end{minipage}
        \begin{center}            
        \text{(d) Nonlinear $\mu_0(x)$, $\mu_1(x)$, nonlinear $e(x)$} 
        \end{center}
        \caption{
        The probability of correctly selecting the better HTE estimator by comparing absolute/relative error's confidence intervals across three methods (plug in, IF, EIF) over four scenarios of the ACIC competition data ((a) to (d)).   
        }
    \label{fig:ACIC.selection.accuracy}
\end{figure}

\section{Discussions}\label{sec:discussion}

The evaluation of HTE estimators typically involves additional estimation steps due to the inherent missingness of counterfactuals. 
We advocate that the comparison of HTE estimators should account for this uncertainty incurred in the test stage to determine the more accurate estimator with statistical confidence. 
We propose to achieve this goal by constructing confidence intervals for the relative error between two HTE estimators rather than their absolute errors.
Explicitly, we derive a one-step correction estimator based on the efficient influence function, and provide the asymptotically narrowest confidence interval of the relative error.
The relative error confidence interval is less sensitive to nuisance function estimation errors and is more powerful in identifying the better HTE estimator when the candidate estimators are similar.
Through empirical evaluation on the benchmark dataset from the 2016 ACIC challenge, we show that our relative error confidence interval achieves robust coverage regardless of the nuisance estimators used, and correctly determines the better HTE estimator with high probability.

We discuss several potential directions for future research. 
\begin{itemize}
    \item In this paper, we focus on evaluating the relative performance of provided HTE estimators. An intriguing direction for future exploration is how the relative error evaluation method can be integrated into the training process to produce a more accurate HTE estimator. Specifically, one potential approach is to use cross-validation in the HTE estimation, where the validation step is carried out using our proposed evaluation method. 

    \item In this paper, we investigate the comparison of HTE estimators' errors averaged across the entire population. However, it's important to note that some HTE estimator could be overall accurate but considerably less precise in underrepresented groups due to insufficient training data. 
    This presents a fairness concern, especially since HTE estimators could influence downstream policy decisions. To address this, it would be beneficial to extend the evaluation to errors averaged over certain subgroups.
    This could also provide insights for improving the HTE estimators within these subgroups.

    \item In this paper, we bypass the degeneracy issue of the asymptotic distribution of the one-step correction estimator by shifting attention from the absolute estimand (absolute error) to its relative counterpart (relative error). It is of interest whether this approach could be generalized to address the degeneracy issue in other scenarios.
\end{itemize}

\printbibliography

\appendix

\section{Notations}\label{appe:sec:notation}
For a causal estimand $\phi$, we use $\IFC(\phi)$ to denote an influence function of $\phi$.
We derive the influence functions assuming the the covariates $x$ are discrete as in \cite{kennedy2022semiparametric}.
We use ${\eta}$ to denote the collection of true nuisance functions, and $\hat{\eta}$ to denote the collection of estimated nuisance functions.
We use $\PP_n$ to denote the distribution of $n$ test data points, and $\EE_n$ be the corresponding expectation (average over $n$ test data points).

\section{Proofs}\label{appe:sec:proof}
\begin{proof}[Proof of \Cref{theo:absolute.error}] 
    We first derive the efficient influence function for the absolute error~\eqref{defi:evaluation.error}. 
    We next prove the asymptotic distribution of the absolute error estimator~\eqref{eq:estimator.absolute.error}.

    \noindent \textbf{Derivation of the efficient influence function}. 
    We consider non-parametric models, where the tangent space contains the entire Hilbert space of mean-zero, finite-variance functions \parencite{tsiatis2006semiparametric}.
    In this case, the influence function is unique, and it must be efficient.
    Below we derive this influence function by applying differentiation rules and using putative influence functions for common causal estimands \parencite{kennedy2022semiparametric}.

    By the linearity of influence functions, 
    \begin{align*}
        \IFC(\phi(\hat{\tau}))
        = \underbrace{\IFC\left(\EE\left[\hat{\tau}^2(X_i)\right]\right)}_{:=\text{(a)}} 
        - 2\underbrace{\IFC\left(\EE\left[\hat{\tau}(X_i) {\tau}(X_i)\right]\right)}_{:=\text{(b)}}
        + \underbrace{\IFC\left(\EE\left[\tau^2(X_i)\right]\right)}_{:=\text{(c)}}.
    \end{align*}
    We deal with (a) to (c) one by one.
    For (a), $\hat{\tau}(x)$ is a known function of the covariates, and the influence function is 
    \begin{align}\label{proof:absolute.error.(a)}
        \text{(a)} = \hat{\tau}^2(X_i) - \EE[\hat{\tau}^2(X_i)].
    \end{align}
    For (b), $\EE\left[\hat{\tau}(X_i) {\tau}(X_i)\right]$ is a weighted average treatment effect with weights $\hat{\tau}(X_i)$, and the influence function is known to be
    \begin{align}\label{proof:absolute.error.(b)}
        \text{(b)} = \hat{\tau}(X_i) \cdot \left(\frac{W_i(Y_i - \mu_1(X_i))}{e(X_i)} + \mu_1(X_i) - \frac{(1-W_i)(Y_i - \mu_0(X_i))}{1-e(X_i)} - \mu_0(X_i) \right) - \EE[\hat{\tau}(X_i) \tau(X_i)].
    \end{align}
    For (c), we rewrite the second moment as a finite sum assuming the discreteness of the covariates,
    \begin{align}\label{proof:absolute.error.(c)}
        \begin{split}
    \text{(c)} 
        &= \sum_x \IFC(\tau^2(x)) \PP_X(dx) +  \sum_x \tau^2(x) \IFC(\PP_X(dx)) \quad (\text{Linearity, product rule of influence functions})\\
        &= \sum_x 2 \tau(x) \IFC(\tau(x)) \PP_X(dx) + \sum_x \tau^2(x) \IFC(\PP_X(dx)) \quad (\text{Chain rule of influence functions}) \\
        &= 2(\mu_1(X) - \mu_0(X)) \left(\frac{W(Y - \mu_1(X))}{e(X)} - \frac{(1-W)(Y - \mu_0(X))}{1-e(X)}\right) \\
        &\quad~ + (\mu_1(X_i) - \mu_0(X_i))^2- \EE[\tau^2(X_i)]. \quad (\text{Influence function of $\tau(x)$, $\PP_X(dx)$.})
        \end{split}
    \end{align}
    Finally, combing (a) to (c) and we have finished the derivation of the influence function for the absolute error~\eqref{defi:evaluation.error}.
       
    \noindent \textbf{Asymptotic distribution of the one-step correction estimator}.
    As in \cite{kennedy2022semiparametric}, we decompose the error into three terms,
    \begin{align*}
    \hat{\phi}(\hat{\tau}) - {\phi}(\hat{\tau})
        &= \EE_n[\hat{\psi}(\phi(\hat{\tau}); Z_i)] + \EE[\hat{\psi}(\phi(\hat{\tau}); Z_i)] - \EE[{\psi}(\phi(\hat{\tau}); Z_i)] \\
        &= \underbrace{\EE_n[{\psi}(\phi(\hat{\tau}); Z_i)] - \EE[{\psi}(\phi(\hat{\tau}); Z_i)]}_{:=\text{S}}  \\
        &\quad~+ \underbrace{\EE_n[\hat{\psi}(\phi(\hat{\tau}); Z_i) - {\psi}(\phi(\hat{\tau}); Z_i)] - \EE[\hat{\psi}(\phi(\hat{\tau}); Z_i) - {\psi}(\phi(\hat{\tau}); Z_i)]}_{:=N} \\
        &\quad~+\underbrace{{\phi}(\hat{\tau}) - \tilde{\phi}(\hat{\tau}) + \EE[\hat{\psi}(\hat{\tau}; Z_i)]}_{:=R},
    \end{align*}
    where $\tilde{\phi}(\hat{\tau})$ denote the absolute error regarding the estimated nuisance functions.
    We deal with $S$, $N$, $R$ one by one.
    For $S$, by the central limit theorem, 
    \begin{align}\label{proof:absolute.error.S}
        \begin{split}
            \sqrt{n} S \stackrel{d}{\to} \calN(0, V(\hat{\phi}(\hat{\tau}))).
        \end{split}
    \end{align}
    For $N$, the empirical process term, by the boundedness assumption ($Y_i$ is bounded, $e(X_i)$ is bounded away from zero and one, and the nuisance function estimators $\hat{\mu}_0(x)$, $\hat{\mu}_1(x)$ are bounded, $\hat{e}(x)$ is bounded away from zero and one), and that nuisance function estimators are consistent in $L_2$ norm, 
    %TODO: add assumptions of the estimated nuisance functions: and the nuisance function estimators $\hat{\mu}_0(x)$, $\hat{\mu}_1(x)$ are bounded, $\hat{e}(x)$ is bounded away from zero and one,
    \begin{align}\label{proof:absolute.error.N}
        n \var(N)
        = \var\left(\hat{\psi}(\phi(\hat{\tau}); Z_i) - {\psi}(\phi(\hat{\tau}); Z_i)\right)
        \to 0.
    \end{align}
    By Chebyshev's inequality, we have $\sqrt{n}N \stackrel{p}{\to} 0$.
    For $R$, the remainder term, we decompose it into the remainder term of (a), (b), and (c) in \eqref{proof:absolute.error.(a)}, \eqref{proof:absolute.error.(b)}, and \eqref{proof:absolute.error.(c)}, respectively.
    $R$(a) is exactly zero. 
    $R$(b) is the remainder term of weighted average treatment effect, and satisfies $o_p(n^{-1/2})$ under the boundedness assumption and the assumption that $\hat{\mu}_1(x)$, $\hat{\mu}_0(x)$, and $\hat{e}_1(x)$ converge at rate $o_p(n^{-1/4})$. 
    For $R$(c),
    \begin{align}\label{proof:absolute.error.R}
        \begin{split}
          R\text{(c)}
          &= -\EE\left[((\mu_1(X_i) - \mu_0(X_i)) - (\tilde{\mu}_1(X_i) - \tilde{\mu}_0(X_i)))^2\right] + 2 \EE\left[(\tilde{\mu}_1(X_i) - \tilde{\mu}_0(X_i)) r(X_i)\right], \\
          &\quad~r(x) := \left( \frac{e(x)}{\tilde{e}(x)} - 1 \right) \left( \mu_1(x) - \tilde{\mu}_1(x) \right) 
          - \left( \frac{1 - e(x)}{1 - \tilde{e}(x)} - 1 \right) \left( \mu_0(x) - \tilde{\mu}_0(x) \right).
        \end{split}
    \end{align}
    % TODO: to check R(c).
    By the boundedness assumption, and the $o_p(n^{-1/4})$ convergence assumption of $\tilde{\mu}_1(x)$, $\tilde{\mu}_0(x)$, and $\tilde{e}(x)$,
    \begin{align*}
        \EE\left[r^2(X_i)\right]
        &\le 2 \EE \left[\left
        (\frac{1 - e(X_i)}{1 - \tilde{e}(X_i)} - 1 \right)^2 \left( \mu_0(X_i) - \tilde{\mu}_0(X_i) \right)^2\right]
        + \EE\left[\left( \frac{1 - e(x)}{1 - \tilde{e}(x)} - 1 \right)^2 \left( \mu_0(x) - \tilde{\mu}_0(x) \right)^2\right]
        = o_p(n^{-1}).
    \end{align*}
    By Cauchy-Schwarz inequality, 
    \begin{align*}
        \EE^2\left[(\tilde{\mu}_1(X_i) - \tilde{\mu}_0(X_i)) r(X_i)\right]
        &\le \EE\left[(\tilde{\mu}_1(X_i) - \tilde{\mu}_0(X_i))^2\right] \cdot  \EE\left[r^2(X_i)\right] 
        = o_p(n^{-1}).
    \end{align*}
    Combining $S$, $N$, $R$ and we have finished the proof of the asymptotic distribution of $\hat{\phi}(\hat{\tau})$ with the true variance ${V}(\hat{\phi}(\hat{\tau}))$.
    As for cross-fitting, we apply the argument in \cite{chernozhukov2018double}.

    For the variance estimator, by the boundedness assumption and the law of large number,
    \begin{align*}
        \hat{V}(\hat{\phi}(\hat{\tau})) \stackrel{p}{\to} {V}(\hat{\phi}(\hat{\tau})).
    \end{align*}
    By the assumption that $\PP(\hat{\tau}(X_i) \neq {\tau}(X_i)) > 0$, we have ${V}(\hat{\phi}(\hat{\tau})) > 0$.
    Finally, by Slutsky's lemma, we have finished the proof of \Cref{theo:absolute.error}.
    
\end{proof}

\begin{proof}[Proof of \Cref{theo:relative.error}]
    Similar to the proof of \Cref{theo:absolute.error},  we first derive the efficient influence function for the relative error~\eqref{defi:evaluation.error.relative}. 
    We next prove the asymptotic distribution of the absolute error estimator~\eqref{eq:estimator.relative.error}.

    \noindent \textbf{Derivation of the efficient influence function}.
    Note that $\delta(\hat{\tau}_1, \hat{\tau}_2) = \phi(\hat{\tau}_1) - \phi(\hat{\tau}_2)$. By the linearity of the influence function, the influence function for $\delta(\hat{\tau}_1, \hat{\tau}_2)$ is simply the difference between the influence functions of $\phi(\hat{\tau}_1)$ and $\phi(\hat{\tau}_2)$ shown in \Cref{theo:absolute.error}. 
    Alternatively, recognize that the estimand $\delta(\hat{\tau}_1, \hat{\tau}_2)$ is essentially a weighted average treatment effect, for which the influence function is a well-established result.

    \noindent \textbf{Asymptotic distribution of the one-step correction estimator}.
    Given that the estimand $\delta(\hat{\tau}_1, \hat{\tau}_2)$ is a weighted average treatment effect, the asymptotic distribution of the one-step correction estimator, under the assumptions of \Cref{theo:relative.error}, follows from the standard result \parencite{kennedy2022semiparametric}. 
    Regarding cross-fitting, we apply the approach outlined in \cite{chernozhukov2018double}. 
    % TODO: add the boundedness assumption.
    
    Additionally, the degeneracy assumption $\mathbb{P}(\hat{\tau}_1(X_i) \neq \hat{\tau}_2(X_i)) > 0$ is required to ensure that the estimator, when scaled by $\sqrt{n}$, has a non-zero variance. Similar to the proof of \Cref{theo:absolute.error}.
\end{proof}

\begin{proposition}\label{prop:relative.error.CI.width}
    There exists $C > 0$ such that
    \begin{align*}
        V(\hat{\delta}(\hat{\tau}_1, \hat{\tau}_2)) \le C \cdot \EE[(\hat{\tau}_1(X) - \hat{\tau}_2(X))^2].
    \end{align*}
    This together with \Cref{theo:relative.error} implies the width of the confidence interval of the relative error estimator~\eqref{eq:estimator.relative.error} based on \Cref{theo:relative.error} is of order $\sqrt{\EE[(\hat{\tau}_1(X) - \hat{\tau}_2(X))^2]} \cdot {n}^{-1/2}$.
\end{proposition}

\begin{proof}[Proof of \Cref{prop:relative.error.CI.width}]
    By Cauchy-Schwarz inequality, 
    \begin{align*}
        V(\hat{\delta}(\hat{\tau}_1, \hat{\tau}_2))
        &\le \EE\left[(\hat{\tau}_2(X_i) - \hat{\tau}_1(X_i))^2\right] \cdot \EE\left[\left(\frac{W_i(Y_i - {\mu}_1(X_i))}{{e}(X_i)} + {\mu}_1(X_i) - \frac{(1-W_i)(Y_i - {\mu}_0(X_i))}{1-{e}(X_i)} - {\mu}_0(X_i)\right)^2\right] \\
        &\le C \cdot \EE\left[(\hat{\tau}_2(X_i) - \hat{\tau}_1(X_i))^2\right].  \quad (\text{Boundedness assumption})
    \end{align*}
\end{proof}

\begin{corollary}[DINA]\label{prop:DINA}
    \Cref{theo:absolute.error} and \Cref{theo:relative.error} apply to the causal estimand in \Cref{rmk:DINA}.
\end{corollary}

\begin{proof}[Proof of \Cref{prop:DINA}]

    To derive the influence functions for the causal estimand in \Cref{rmk:DINA}, we can apply the chain rule to the influence functions obtained in the proofs of \Cref{theo:absolute.error} and \Cref{theo:relative.error}. 

    For the asymptotic distribution of the one-step correction estimator, we can apply the arguments in \Cref{theo:absolute.error} and \Cref{theo:relative.error}.
\end{proof}

\section{Additional empirical experiments}\label{appe:sec:simulation.additional}
We provide additional demonstrations of the hypothetical real data analysis in \Cref{sec:simulation}. The datasets, HTE estimators, and evaluation methods remain the same. 
In \Cref{appe:sec:simulation.width}, we present the widths of the confidence intervals, which can serve as an indication of the statistical power (a shorter confidence intervals suggests higher power). 
In \Cref{appe:sec:simulation.estimation.error}, we report the estimation errors for both the absolute and relative error estimators.

\subsection{Widths of confidence intervals}\label{appe:sec:simulation.width}

In \Cref{fig:ACIC.width}, we report the widths of the estimated absolute (LASSO, Boosting) and relative (LASSO vs. Boosting) error 90\% confidence intervals across three methods (plug-in, IF, EIF) over four scenarios from the ACIC competition data ((a) to (d)).
Shorter widths indicate higher power of identifying the better HTE estimator.

We make the following observations:
\begin{itemize}
    \item Across all settings and methods, the confidence intervals for the relative errors are significantly shorter than those for the absolute errors.
    \item For the absolute error estimators, our proposal performs significantly better than the estimator in \cite{alaa2019validating} in scenario (d), and the two methods are comparable in scenarios (a) to (c). The plug-in estimator produces the widest intervals across all settings.
\end{itemize}

\begin{figure}[ht]
        \centering
        \begin{minipage}{0.3\textwidth}
                \centering
                \includegraphics[clip, trim = 0cm 0cm 0cm 0cm, width = \textwidth]{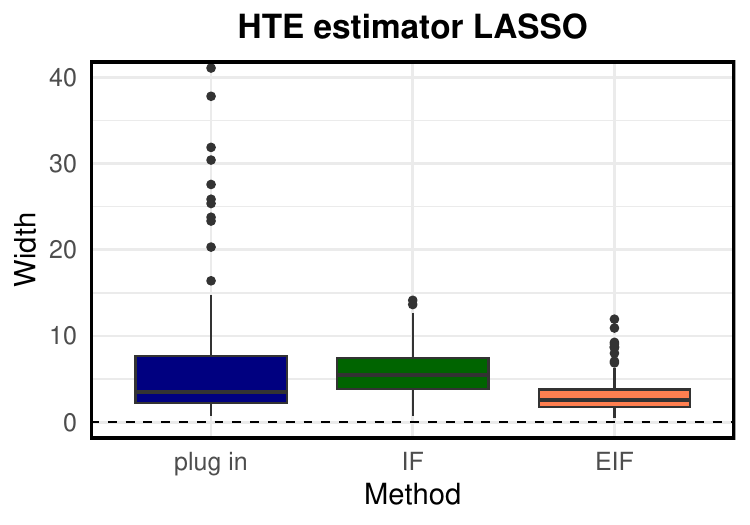}
        \end{minipage}
        \begin{minipage}{0.3\textwidth}
                \centering
                \includegraphics[clip, trim = 0cm 0cm 0cm 0cm, width = \textwidth]{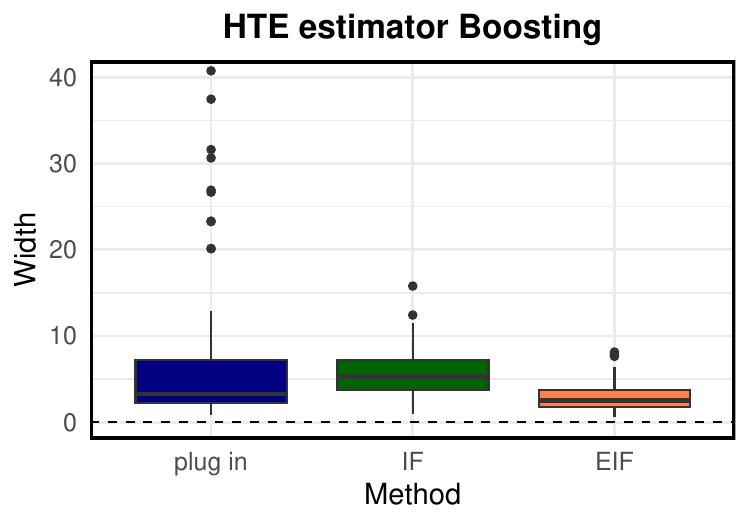}
        \end{minipage}
        \begin{minipage}{0.3\textwidth}
                \centering
                \includegraphics[clip, trim = 0cm 0cm 0cm 0cm, width = \textwidth]{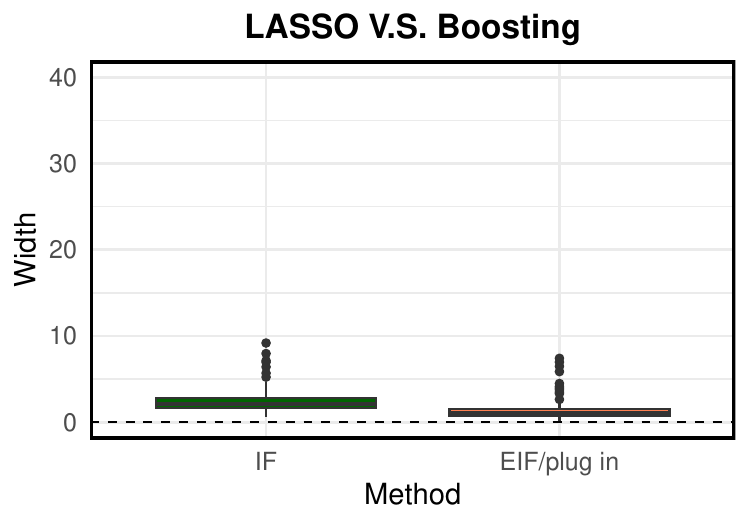}
        \end{minipage}        
        \\
        \begin{center}
        \text{(a) Linear $\mu_0(x)$, $\mu_1(x)$, linear $e(x)$}\\
        \end{center}
        \begin{minipage}{0.3\textwidth}
                \centering
                \includegraphics[clip, trim = 0cm 0cm 0cm 0cm, width = \textwidth]{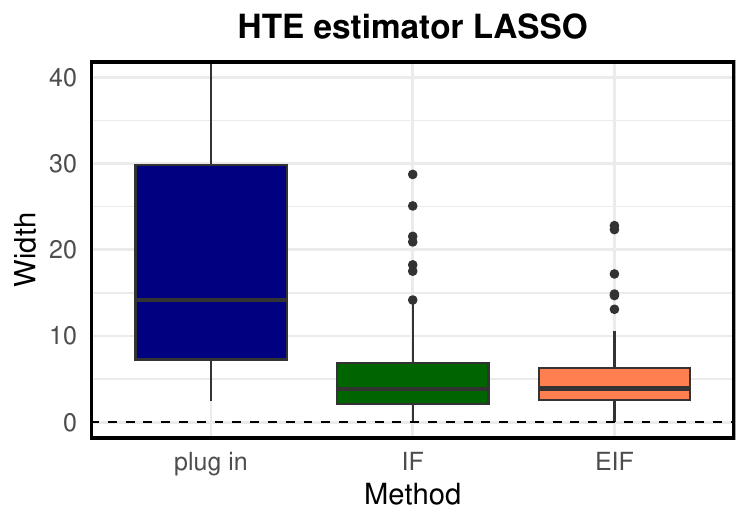}
        \end{minipage}
        \begin{minipage}{0.3\textwidth}
                \centering
                \includegraphics[clip, trim = 0cm 0cm 0cm 0cm, width = \textwidth]{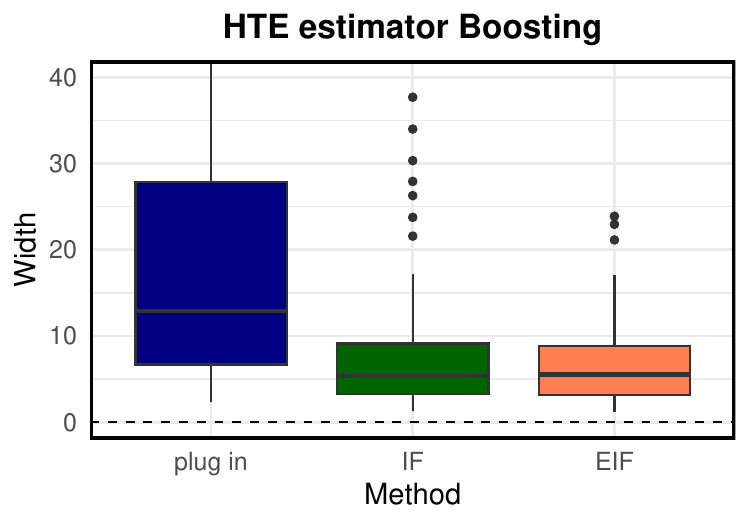}
        \end{minipage}
        \begin{minipage}{0.3\textwidth}
                \centering
                \includegraphics[clip, trim = 0cm 0cm 0cm 0cm, width = \textwidth]{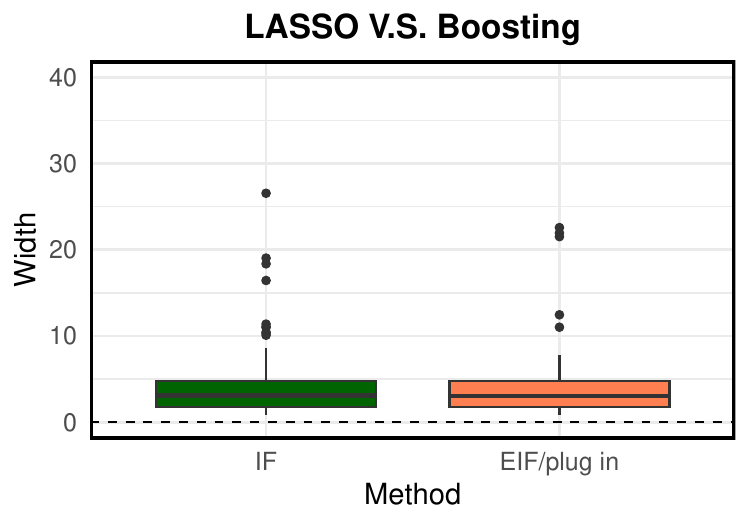}
        \end{minipage}     \\
                \begin{center}
        \text{(b) Nonlinear $\mu_0(x)$, $\mu_1(x)$, linear $e(x)$}  
        \end{center}
        \begin{minipage}{0.3\textwidth}
                \centering
                \includegraphics[clip, trim = 0cm 0cm 0cm 0cm, width = \textwidth]{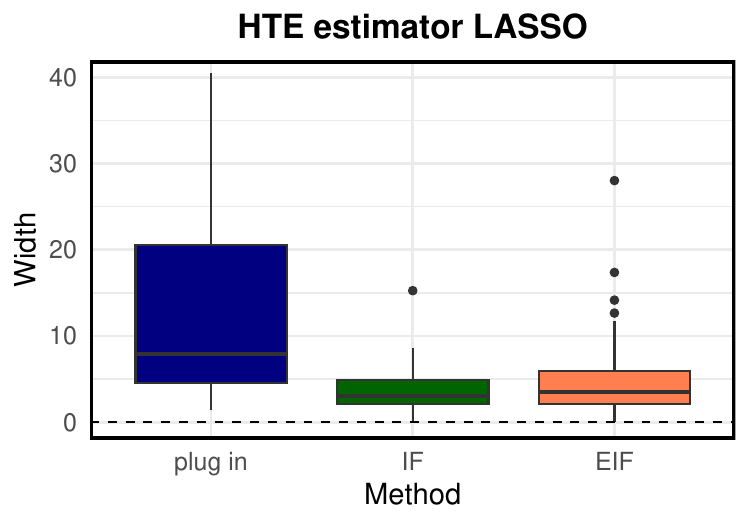}
        \end{minipage}
        \begin{minipage}{0.3\textwidth}
                \centering
                \includegraphics[clip, trim = 0cm 0cm 0cm 0cm, width = \textwidth]{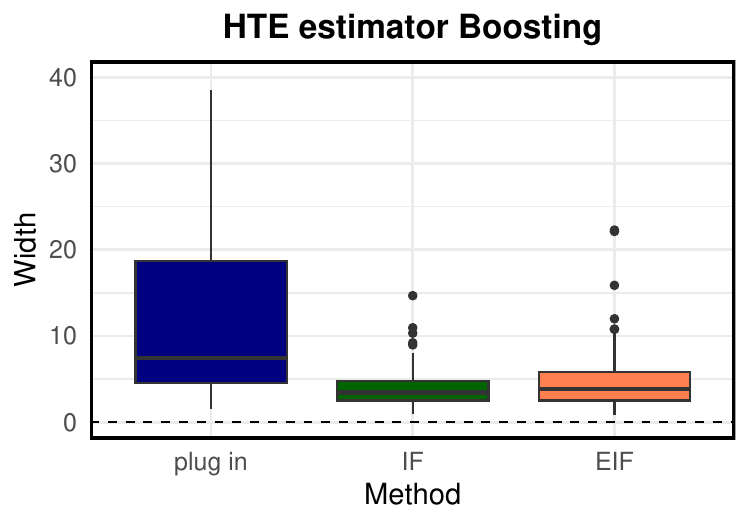}
        \end{minipage}
        \begin{minipage}{0.3\textwidth}
                \centering
                \includegraphics[clip, trim = 0cm 0cm 0cm 0cm, width = \textwidth]{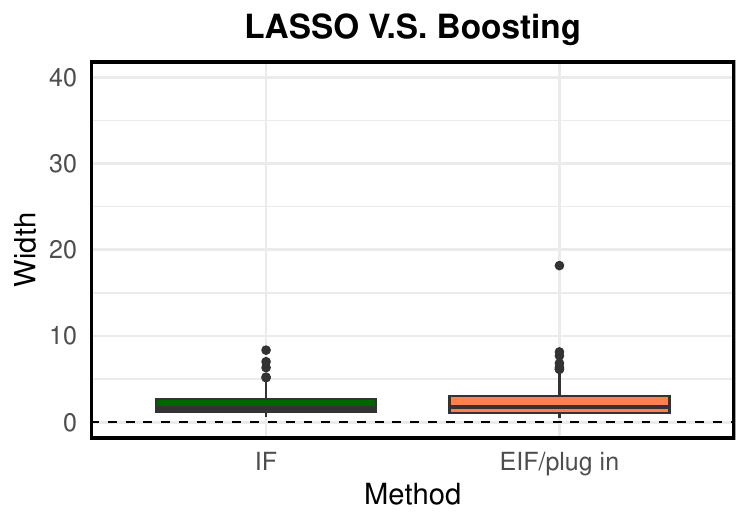}
        \end{minipage}        \\
                \begin{center}
        \text{(c) Linear $\mu_0(x)$, $\mu_1(x)$, nonlinear $e(x)$}      
        \end{center}
        \begin{minipage}{0.3\textwidth}
                \centering
                \includegraphics[clip, trim = 0cm 0cm 0cm 0cm, width = \textwidth]{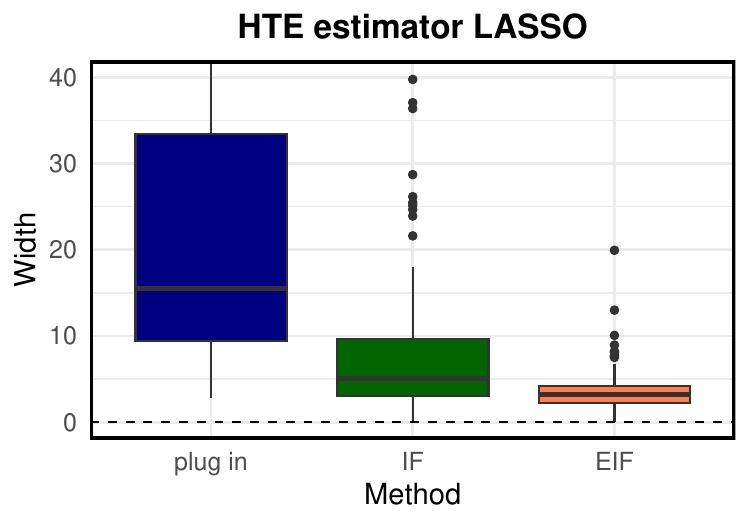}
        \end{minipage}
        \begin{minipage}{0.3\textwidth}
                \centering
                \includegraphics[clip, trim = 0cm 0cm 0cm 0cm, width = \textwidth]{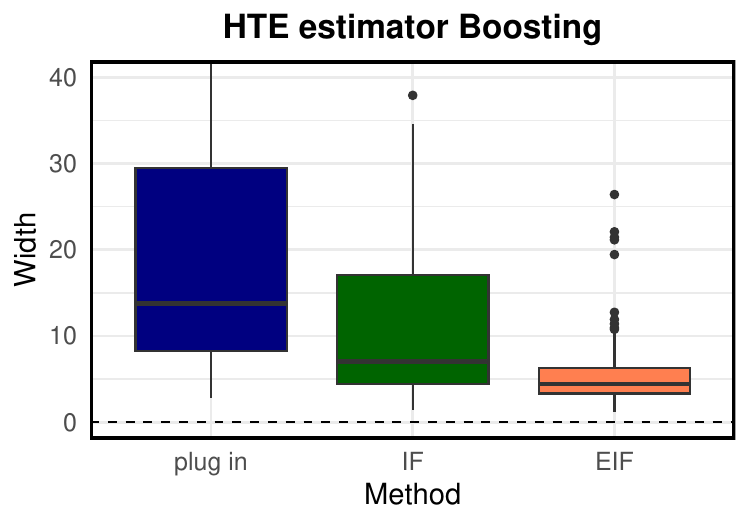}
        \end{minipage}
        \begin{minipage}{0.3\textwidth}
                \centering
                \includegraphics[clip, trim = 0cm 0cm 0cm 0cm, width = \textwidth]{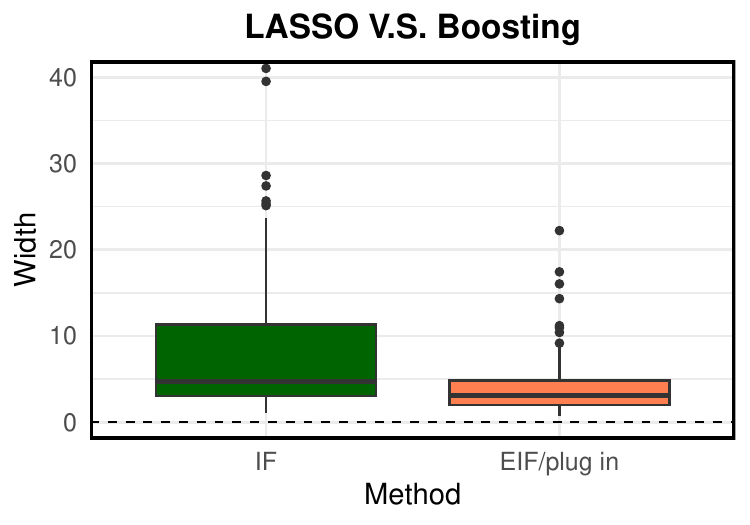}
        \end{minipage}        \\
                \begin{center}
        \text{(d) Nonlinear $\mu_0(x)$, $\mu_1(x)$, nonlinear $e(x)$} 
        \end{center}
        \caption{
        Width of the estimated absolute (LASSO, Boosting)/relative (LASSO V.S. Boosting) error's $90\%$ confidence intervals across three methods (plug in, IF, EIF) over four scenarios of the ACIC competition data ((a) to (d)). 
        }
    \label{fig:ACIC.width}
\end{figure}

\subsection{Error of estimated errors}\label{appe:sec:simulation.estimation.error}

Error of the estimated absolute (LASSO, Boosting)/relative (LASSO V.S. Boosting) error across three methods (plug in, IF, EIF) over four scenarios of the ACIC competition data ((a) to (d)). 
Error of the estimated error is defined as the absolute difference between the estimated error and the corresponding oracle value.
Smaller errors suggest that the estimator has higher accuracy thus more favorable.

We make the following observations.
\begin{itemize}
    \item Across all settings and methods, the errors of the relative error estimators are significantly smaller than those for the absolute errors.
    \item For the absolute error estimators, our proposal performs significantly better than the estimator in \cite{alaa2019validating} and the plug-in estimator in scenario (d). The plug-in estimator is unfavorable from scenarios (b) to (d).
\end{itemize}

\begin{figure}[ht]
        \centering
        \begin{minipage}{0.3\textwidth}
                \centering
                \includegraphics[clip, trim = 0cm 0cm 0cm 0cm, width = \textwidth]{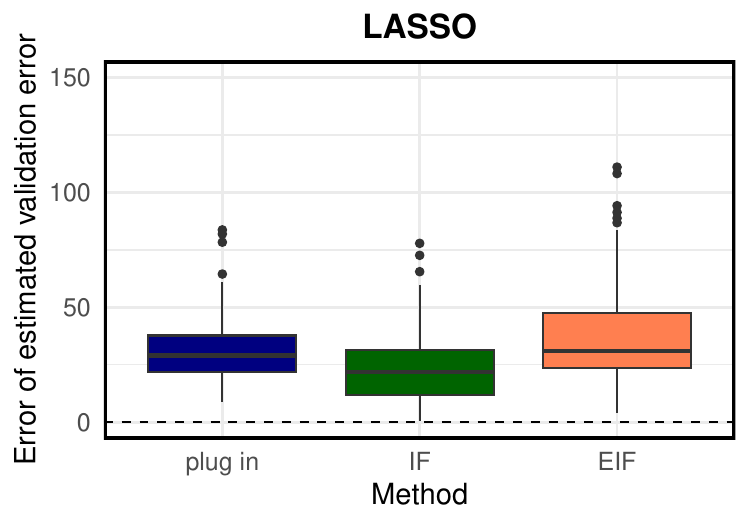}
        \end{minipage}
        \begin{minipage}{0.3\textwidth}
                \centering
                \includegraphics[clip, trim = 0cm 0cm 0cm 0cm, width = \textwidth]{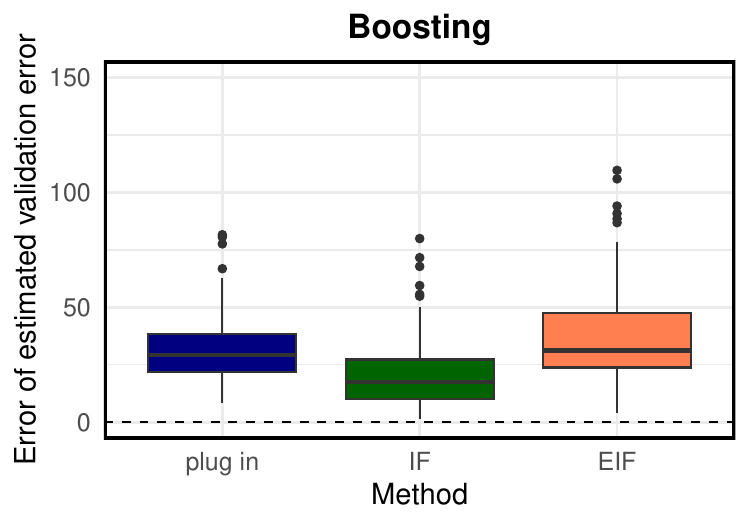}
        \end{minipage}
        \begin{minipage}{0.3\textwidth}
                \centering
                \includegraphics[clip, trim = 0cm 0cm 0cm 0cm, width = \textwidth]{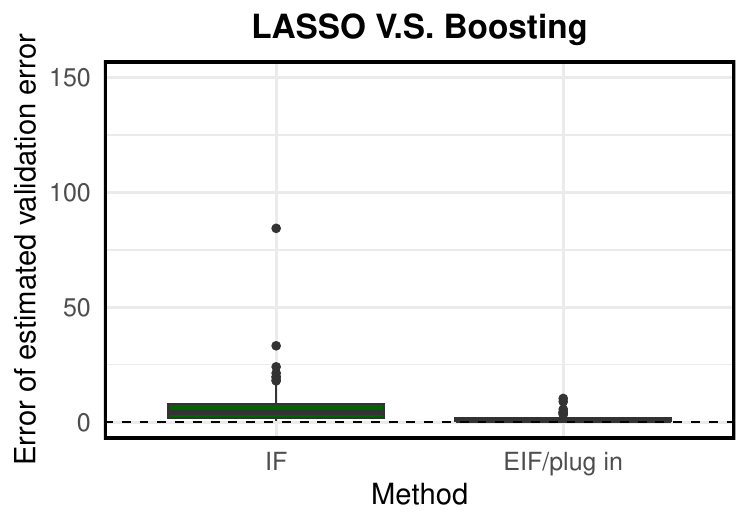}
        \end{minipage}   
        \\
        \begin{center}
        \text{(a) Linear $\mu_0(x)$, $\mu_1(x)$, linear $e(x)$}
        \\            
        \end{center}
        \begin{minipage}{0.3\textwidth}
                \centering
                \includegraphics[clip, trim = 0cm 0cm 0cm 0cm, width = \textwidth]{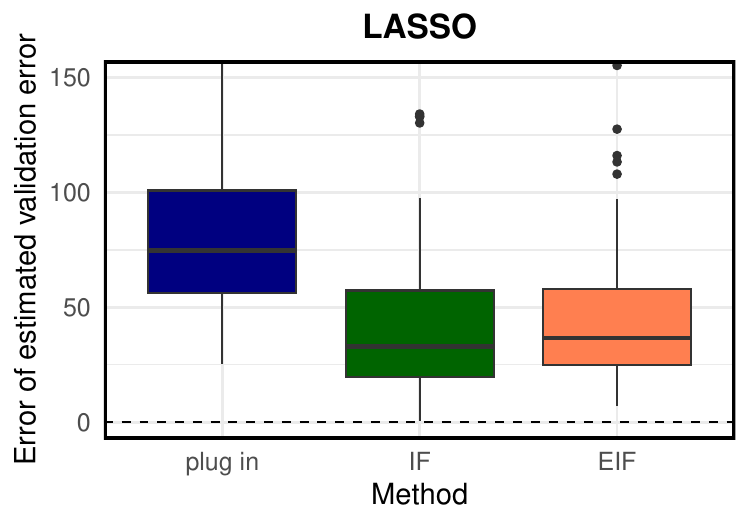}
        \end{minipage}
        \begin{minipage}{0.3\textwidth}
                \centering
                \includegraphics[clip, trim = 0cm 0cm 0cm 0cm, width = \textwidth]{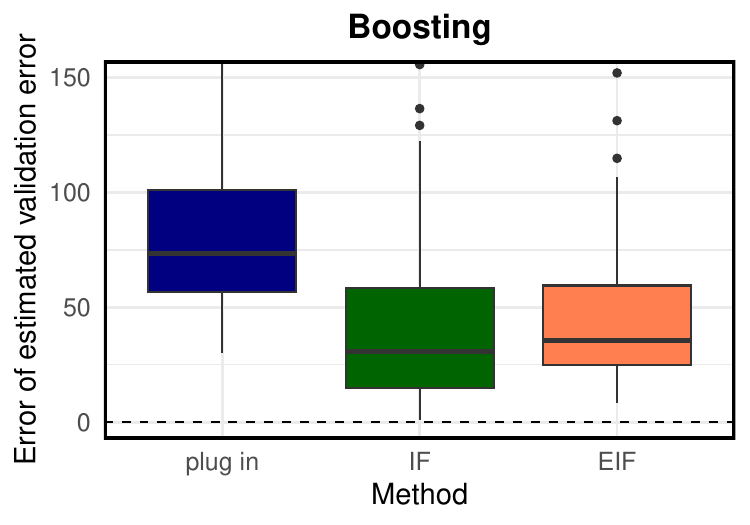}
        \end{minipage}
        \begin{minipage}{0.3\textwidth}
                \centering
                \includegraphics[clip, trim = 0cm 0cm 0cm 0cm, width = \textwidth]{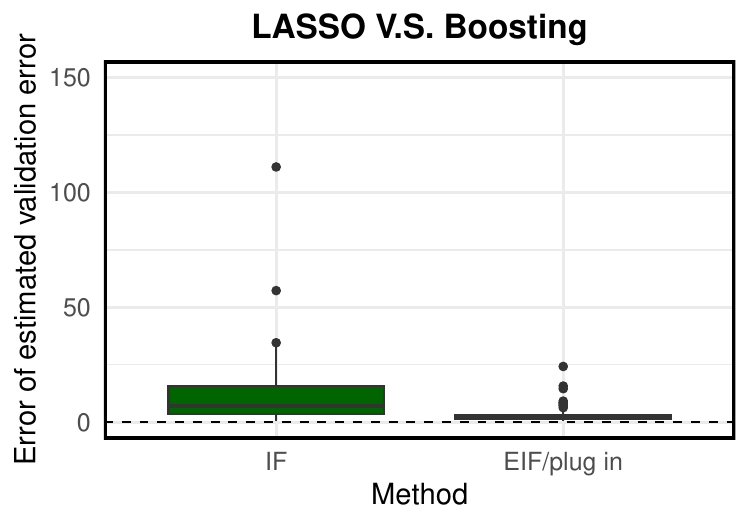}
        \end{minipage}     \\   
        \begin{center}        
        \text{(b) Nonlinear $\mu_0(x)$, $\mu_1(x)$, linear $e(x)$} \\ 
        \end{center}
        \begin{minipage}{0.3\textwidth}
                \centering
                \includegraphics[clip, trim = 0cm 0cm 0cm 0cm, width = \textwidth]{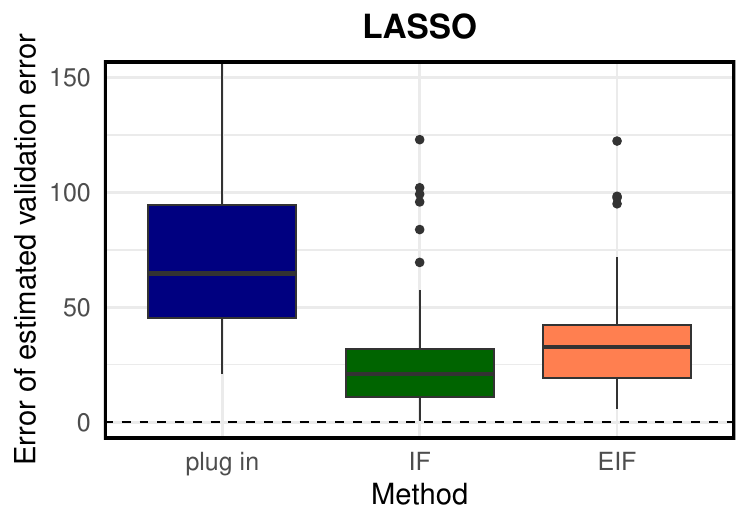}
        \end{minipage}        
        \begin{minipage}{0.3\textwidth}
                \centering
                \includegraphics[clip, trim = 0cm 0cm 0cm 0cm, width = \textwidth]{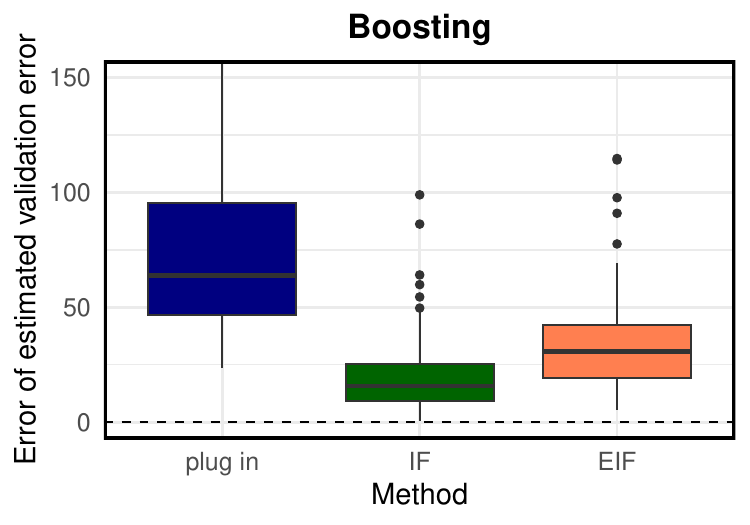}
        \end{minipage}
        \begin{minipage}{0.3\textwidth}
                \centering
                \includegraphics[clip, trim = 0cm 0cm 0cm 0cm, width = \textwidth]{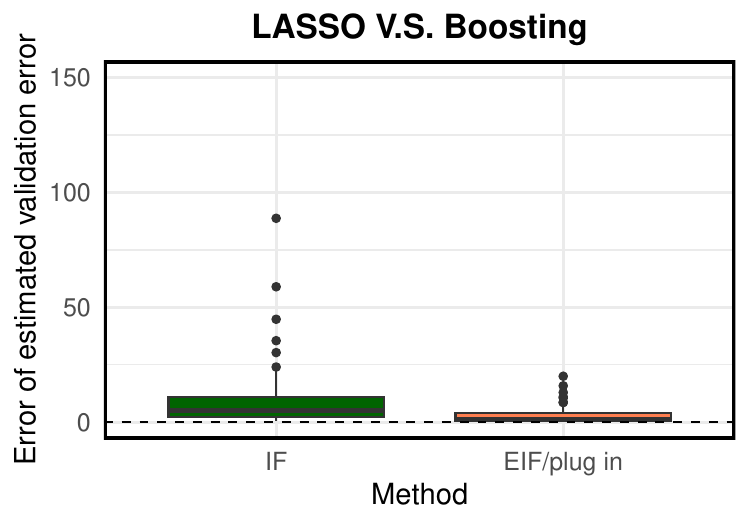}
        \end{minipage}\\        
        \begin{center}
        \text{(c) Linear $\mu_0(x)$, $\mu_1(x)$, nonlinear $e(x)$}   \\   
        \end{center}
        \begin{minipage}{0.3\textwidth}
                \centering
                \includegraphics[clip, trim = 0cm 0cm 0cm 0cm, width = \textwidth]{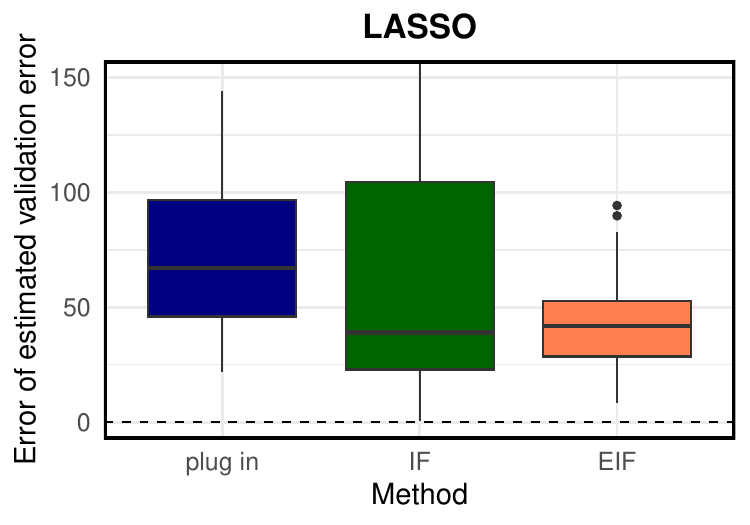}
        \end{minipage}
        \begin{minipage}{0.3\textwidth}
                \centering
                \includegraphics[clip, trim = 0cm 0cm 0cm 0cm, width = \textwidth]{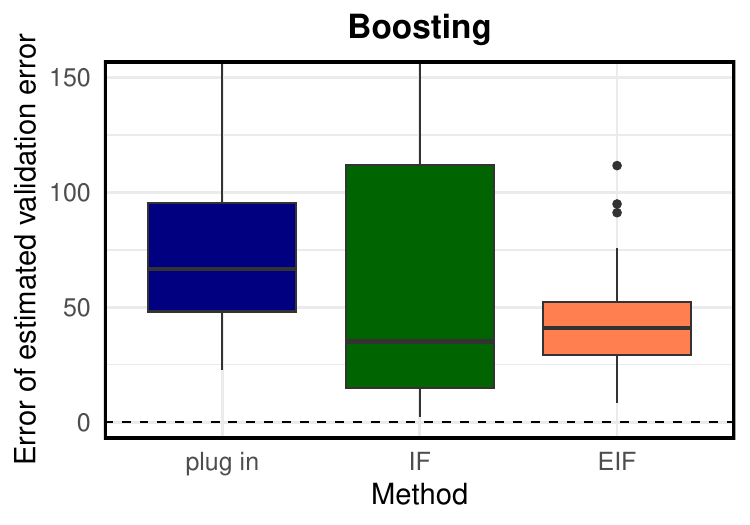}
        \end{minipage}
        \begin{minipage}{0.3\textwidth}
                \centering
                \includegraphics[clip, trim = 0cm 0cm 0cm 0cm, width = \textwidth]{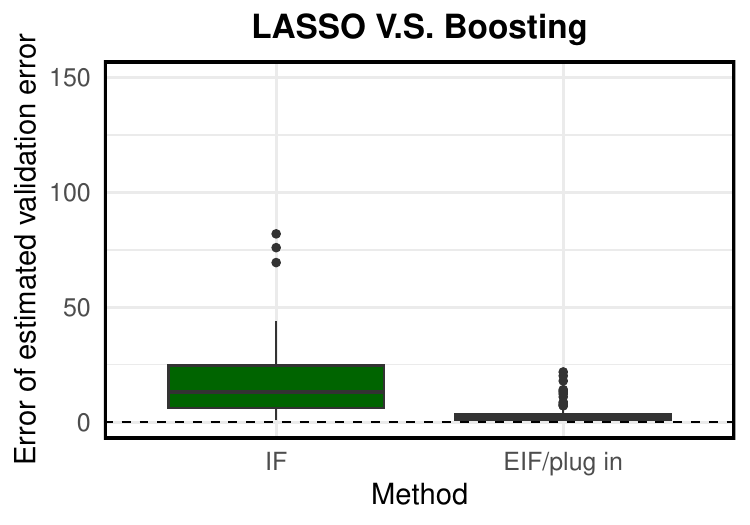}
        \end{minipage}        \\
        \begin{center}
        \text{(d) Nonlinear $\mu_0(x)$, $\mu_1(x)$, nonlinear $e(x)$} \\
        \end{center}
        \caption{
        Error of the estimated absolute (LASSO, Boosting)/relative (LASSO V.S. Boosting) error across three methods (plug in, IF, EIF) over four scenarios of the ACIC competition data ((a) to (d)). 
        Error of the estimated error is defined as the absolute difference between the estimated error and the corresponding oracle value.
        }
    \label{fig:ACIC.error.estimated.error}
\end{figure}

\end{document}